\documentclass[conference,compsoc]{IEEEtran}
\IEEEoverridecommandlockouts
\usepackage{algorithm}
\usepackage{algpseudocode}
\usepackage{algorithmicx}
\usepackage{amsmath}
\usepackage{amsthm}
\usepackage{amssymb}
\usepackage{mathtools}
\usepackage{bbm}
\usepackage{subfigure}

\usepackage{ulem}

\newtheorem{thm}{Theorem}
\newtheorem{defi}{Definition}
\newtheorem{prop}{Proposition}

\newtheorem{rmk}{Remark}

\usepackage{tikz}
\usepackage{amsmath}

\ifCLASSOPTIONcompsoc
  \usepackage[nocompress]{cite}
\else
  \usepackage{cite}
\fi
\ifCLASSINFOpdf
\else
\fi
\begin{document}
%
\title{Budget Recycling Differential Privacy\thanks{This paper has been accepted for 45th IEEE Symposium on
Security and Privacy}}



%
\author{\IEEEauthorblockN{Bo Jiang\IEEEauthorrefmark{1},
Jian Du\IEEEauthorrefmark{1},
Sagar Sharma\IEEEauthorrefmark{1} and
    Qiang Yan\IEEEauthorrefmark{1}}
\IEEEauthorblockA{\IEEEauthorrefmark{1}
TikTok Inc\\ Email: \{bojiang, jian.du, sagar.sharma, yanqiang.mr\}@tiktok.com}}


\maketitle

\begin{abstract}
Differential Privacy (DP) mechanisms usually {force} reduction in data utility by producing ``out-of-bound'' noisy results for a tight privacy budget. We introduce the Budget Recycling Differential Privacy (BR-DP) framework, designed to provide soft-bounded noisy outputs for a broad range of existing DP mechanisms. By ``soft-bounded," we refer to the mechanism's ability to release most outputs within a predefined error boundary, thereby improving utility and maintaining privacy simultaneously.
The core of BR-DP consists of two components: a DP kernel responsible for generating a noisy answer per iteration, and a recycler that probabilistically recycles/regenerates or releases the noisy answer. We delve into the privacy accounting of BR-DP, culminating in the development of a budgeting principle that optimally sub-allocates the available budget between the DP kernel and the recycler.
Furthermore, we introduce algorithms for tight BR-DP accounting in composition scenarios, and our findings indicate that BR-DP achieves reduced privacy leakage post-composition compared to DP.
Additionally, we explore the concept of privacy amplification via subsampling within the BR-DP framework and propose optimal sampling rates for BR-DP across various queries. We experiment with real data, and the results demonstrate BR-DP's effectiveness in lifting the utility-privacy tradeoff provided by DP mechanisms.
\end{abstract}


%
\IEEEpeerreviewmaketitle

\section{Introduction}
In today's data-driven landscape, service providers collecting sensitive user data are compelled to navigate the delicate balance between privacy protection and data utility, particularly in light of stringent privacy regulations such as the General Data Protection Regulation (GDPR)\cite{GDPR2016a}. Differential Privacy (DP)\cite{Dwork2006} has emerged as the \textit{de facto} standard in this context, adept at safeguarding individual privacy within datasets while maintaining the integrity of group-level statistics and insights. Also, DP has led to multiple real-world applications such as Apple's data collection and machine learning \cite{apple_learning_privacy}, Google's RAPPOR \cite{Rappor}, 2020 U.S. Census  \cite{inproceedingsCensus}, etc.
DP stands out for its ability to provide robust privacy guarantees against even the most knowledgeable adversaries, who may possess extensive information about individuals in the dataset except for the targeted individual. This protection is rooted in a well-defined privacy mechanism that effectively masks the individual's contribution to aggregated outcomes.
Central to DP's mechanism is the concept of a `privacy budget', which serves to measure the strictness of the privacy guarantee provided and can also be understood as the maximum possible privacy leakage caused by a DP mechanism. A small budget implies a stronger privacy guarantee with minimal leakage, and hence is typically preferable. 

Differential privacy is typically achieved by introducing randomness into aggregated data results through various mechanisms. Notably, the Laplace and Gaussian mechanisms add random noise from their respective distributions to aggregation results such as sum, average, and count\cite{DP_mec}. Other mechanisms like the exponential and randomized response \cite{wang2016using} select an item from a candidate list based on probabilistic criteria. Discrete Gaussian\cite{10.5555/3495724.3497039} and Binomial mechanisms \cite{10.5555/3327757.3327856} extend the DP applications for discrete/categorical data aggregation.
While this randomness provides deniability for each individual to guarantee privacy protection, it inevitably leads to a reduction in data utility, which is typically quantified by measuring the deviation between the true aggregate and the noise-infused output.  Moreover, in real-world settings, these mechanisms might release out-of-bound results that are unacceptable by analysts under small privacy budgets. These ``error bounds" are critically essential in numerous real-world applications. For instance, in scenarios requiring non-negative aggregate results, outputs less than zero are inherently invalid\cite{holohan2018bounded}. When measuring aggregates involving a small user base, the acceptable range of the noisy aggregates has to be quite narrow to provide meaningful user insights\cite{10.1145/1989323.1989348}. Similarly, in A/B testing, results that erroneously reverse the expected ordering are considered unacceptable\cite{10.1145/3487553.3524216}. 


In the realm of DP research, various attempts have been proposed to tackle the challenge of bounding DP error. Among these, truncated or bounded DP mechanisms stand out as significant innovations. The truncated mechanism operates by clipping the probability density function of the output to a valid range, reassigning out-of-bound values to their nearest acceptable limits\cite{gupte2010universally}. Truncated mechanisms trivially satisfy DP due to the post-processing manner, however, they lead to a high density of noisy responses on the error boundaries. Conversely, the bounded mechanism combines truncation with normalization\cite{Geng2018TruncatedLM, 2022-2-7-0193, holohan2018bounded}, and therefore, providing a monotonically decaying probability density function (PDF). However, bounded mechanisms introduced in literature still possess several unresolved issues:
Firstly, most of these mechanisms are constrained to a fixed range of output support, which falls short in addressing the broader ``bounded error'' challenge that requires the noisy magnitude to be upper bounded. Moreover, the introduction of bias in the released data compromises stability and reliability. Secondly, the necessity to bound noise magnitude inevitably leads to a high failure probability\cite{Geng2018TruncatedLM}. This is because the discrepancy in noise support between neighboring datasets, which scales linearly with query sensitivity, results in unbounded privacy leakage. As a result, for small $\delta$,  bounded noisy magnitude mechanisms are infeasible. Finally, accurately measuring privacy leakage after multiple queries (termed sequential composition) is particularly challenging with these mechanisms due to their ``irregular'' and typically asymmetric geometry of noise distributions. It is worth noting that 
composability is a fundamental attribute of DP mechanisms, pivotal for determining their validity and compliance with privacy regulations. A series of studies have focused on refined composition accountant, employing methodologies like moment accounting\cite{10.1145/2976749.2978318}, which has been expanded into Renyi DP\cite{8049725} and zero Concentrated DP\cite{bun2016concentrated}, or those based on the Privacy Loss Distribution (PLD) to numerically ascertain precise privacy leakage\cite{Balle2018ImprovingTG,Wang2019,zhu2019poission,mironov2019r}. However, applying these advanced techniques to bounded mechanisms proves to be infeasible, presenting a notable gap in current DP research.

It is noteworthy that while bounded mechanisms offer stringent utility guarantees, these may not always align with practical necessities. {E.g., developers are usually satisfied with aggregated results that falling in given error bounds with statistical guarantee. }We investigate a more pragmatic yet fundamental question in tackling the error-bound challenge: Can we redesign DP mechanisms to enhance the likelihood of noisy outputs falling within an acceptable error range? Considering the existing body of research on DP composition accounting, furthermore, a pivotal question arises: Can we accurately quantify the compositional privacy leakage of our proposed mechanism, utilizing established tools in this domain? To address these critical inquiries, this paper introduces the Budget Recycling DP (BR-DP) framework. This framework innovatively integrates conventional DP mechanisms, such as the Laplace or Gaussian mechanism, with a budget recycling phase. Conceptually, the framework optimally splits a portion of the total available privacy budget, to generate a noisy version \( y_n \) of the target query. Concurrently, the remaining budget is allocated to the recycler, which conditionally releases the result based on the acceptability of the noise magnitude. In scenarios where the noise magnitude exceeds the tolerable range, the recycler, governed by a probabilistic rule, either redirects the process to regenerate another noisy result or opts to release the current result despite its unacceptability. This iterative cycle continues until an acceptable noisy version of the query result is produced.

Our main contributions of this paper are as follows:
\begin{enumerate}
    \item We introduce BR-DP framework, which integrates a DP kernel mechanism and a recycler. This framework achieves DP while providing enhanced utility, as measured by the likelihood of generating acceptable outputs within error bounds. BR-DP is adaptable to various specific DP mechanisms, including Gaussian and Laplacian.
    
    \item Our analysis offers a comprehensive examination of privacy leakage. Based on this, we develop optimal budget allocation principles. These principles balance utility maximization and DP compliance for any given privacy budget, distributing it effectively between the DP kernel and the recycler.
    
    \item We propose a tight BR-DP composition theorem and an accounting algorithm characterized by linear complexity. Comparative results indicate that BR-DP achieves lower privacy leakage than the conventional DP mechanism post-composition under the same privacy budget per query.
    
    \item By integrating BR-DP with privacy amplification via subsampling, we formulate an optimal sampling rate determination algorithm for various query types. This enhances utility, particularly under constrained privacy budgets.
    
    \item Our evaluation, using real-world datasets, demonstrates that BR-DP significantly improves data utility across diverse query types with less privacy leakage accounted for after composed usages.

\end{enumerate}

\section{Preliminaries of Differential Privacy}
In this section, we introduce DP along with some related properties and some related state-of-the-art techniques. Then we introduce the utility-privacy tradeoff of DP where the utility is defined with an accuracy guarantee.

\subsection{DP and Accounting Techniques}
DP\cite{Dwork2006,Dwork2008,Dwork20061} provides a mathematical guarantee that the presence or absence of an individual's data in a dataset does not significantly affect the outcome of queries. It is formally defined as:

\begin{defi}
A randomized algorithm $\mathcal{M}$ that takes as input a dataset consisting of individuals is $(\epsilon,\delta)$-differentially private (DP) if for any pair of datasets $X'$, $X$ that differ in the record of a single individual, and any event $S$:
\begin{equation}\label{DP:def}
    \operatorname{Pr}(\mathcal{M}(X) \in S) \leq e^\epsilon \cdot \operatorname{Pr}(\mathcal{M}(X') \in S)+\delta.
\end{equation}
When $\delta$ = 0, the guarantee is simply called (pure) $\epsilon$-DP.
\end{defi}
In \eqref{DP:def}$, \epsilon$ represents the privacy budget; $\delta\in[0,1]$ denotes a failure probability that $\mathcal{M}$ failed to provide $\epsilon$-DP protection. Small $\epsilon$ and $\delta$ imply strict privacy guarantees.

In an interactive Differential Privacy (DP) framework, a trusted server is presumed to store sensitive datasets. A data analyst, keen on a specific aggregated function over dataset $X$ held by the server, dispatches a query $Q(\cdot)$ to it. To address privacy considerations, the server utilizes a differential privacy mechanism upon computing the raw response, ensuring the confidentiality of $X$. Consequently, depending on the query sensitivity $\Delta_f$, noise $N$ is introduced proportionally by the DP mechanism to the raw answer, resulting in the release of $Y_{n}=Y+N$\footnote{We employ capitalized letters to denote random variables, and lowercase mathematical symbols to represent their corresponding realizations.}. 

Depending on the distribution of $N$, DP mechanisms can be classified into many categories, and the most classic mechanisms achieving DP are Gaussian and Laplacian \cite{10.1561/0400000042, DP_mec}. They each possess advantages in nature, e.g. Laplacian mechanism can achieve pure $\epsilon$-DP without introducing $\delta$ and is geometrically more concentrated around the mean. The Gaussian mechanism is more natural in composition (introduced later) and achieves small leakage after multiple uses, also the standard deviation of the
noise is proportional to the query’s $l2$ sensitivity, which is no
larger and often much smaller than $l1$.
These noisy distributions can be represented by some parameters, such as scale $b$ for Laplacian distribution, and standard deviation $\sigma$ for Gaussian distribution (DP mechanisms are typically 0-mean). A fundamental question for DP is how to optimize these parameters representing a noisy distribution to achieve $(\epsilon,\delta)$-DP. This topic is usually referred to as DP accounting. In the original work of \cite{DP_mec}, closed-form values of $b$ and $\sigma$ are derived to achieve $(\epsilon,\delta)$-DP. However, these relationships stem from the loose bound of the DP leakage, and the actual leakage is much smaller. Therefore these initial results are no longer favorable nowadays. In \cite{Balle2018ImprovingTG}, Balle et al discovered several limitations on the Gaussian mechanism and proposed an analytic Gaussian mechanism that can be used for tight accounting. \cite{10.5555/3495724.3497039} proposed accounting algorithm for discrete Gaussian mechanism. 

Another important property of DP is its composability, which can be roughly defined as how to analyze the cumulative privacy guarantee in terms of $\epsilon$ and $\delta$ when multiple queries are performed on a dataset. Naive composition theorem suggests that $\epsilon$ and $\delta$ both increase linearly with the number of queries \cite{Dwork2006} (or sublinearly with advanced composition \cite{kairouz2015composition}). The large budget consumption makes the DP mechanism hard to implement in the real-world, especially those machine learning applications, which require multiple iterations of gradient updates. On the other hand, as DP algorithms have become more intricate, the task of precisely quantifying cumulative privacy loss has similarly grown in complexity. A significant advancement in this area was the introduction of the moments accountant method by Abadi et al \cite{10.1145/2976749.2978318}. This method greatly enhanced the accuracy of privacy loss estimates in compositions involving subsampled Gaussian mechanisms, which are prevalent in DP stochastic gradient descent (DP-SGD). Subsequent enhancements have been made, notably in the realm of Rényi Differential Privacy (RDP), and Concentrated differential privacy (CDP) as pioneered by Mironov\cite{8049725} and Bun\cite{bun2016concentrated}. Further refinements include more stringent RDP bounds for subsampled mechanisms, contributed by researchers such as \cite{Balle2018ImprovingTG,Wang2019,zhu2019poission,mironov2019r}. While RDP provides a rigorous framework for analyzing compositions involving Gaussian mechanisms, its applicability to other mechanisms remains challenging. Furthermore, converting RDP findings to the more widely recognized $(\epsilon,\delta)$ privacy guarantees often entails a loss of precision. Another composition accounting direction is based on privacy loss distribution (PLD), introduced by Sommer et al.\cite{cryptoeprint:2018/820}. Then Koskela et al \cite{Koskela2019ComputingTD, Koskela2020TightAD} proposes to use FFT-based algorithms that take the PLD as a time series signal, and numerically calculate the cumulative leakage in the frequency domain. The latest work along this direction is in \cite{Zhu2021OptimalAO}, where Yuqing {et} el. proposes an analytic Fourier accounting algorithm deploying the characteristic function, which overcomes a shortage of the FFT-based algorithm that the worst-case PLD may involve exhaustively searching for all neighboring datasets.


\subsection{Utility-privacy Tradeoff}

DP mechanisms typically concentrate on minimizing the distance between the noisy release and the true answer. Accordingly, the utility definitions employed in DP literature are generally based on the absolute distance measure.

\textbf{Absolute Distance:}  Given \( Y_n = Y + N \), the absolute error can be regarded as dependent solely on the magnitude of the \( l \)-norm noise:  
\begin{equation*}
    ||N||_l,
\end{equation*}
where \( l = 1 \) corresponds to the Manhattan distance, \( l = 2 \) to the Euclidean distance, and \( l = \infty \) typically refers to the Chebyshev distance. In the following, we directly use $||\cdot||$ for simplicity. 



Depending on the application of DP, to ensure the quality of service, typically, a noisy output $y_n$ that deviates too much from $y$ is not \textit{acceptable}. Therefore, DP service providers usually require a certain utility guarantee. i.e., a statistically bounded error. Mathematically:
\begin{equation}\label{eq:utility}
    \text{Pr}(||N||\le{\theta})\ge{\rho},
\end{equation}
where $\theta$ represents the upper error boundary,  and $\rho$ stands for the confidence. The probability defined in \eqref{eq:utility} is referred to as \textit{acceptance} rate.

\begin{figure*}[t]
    \centering
    \includegraphics[width=0.7\textwidth]{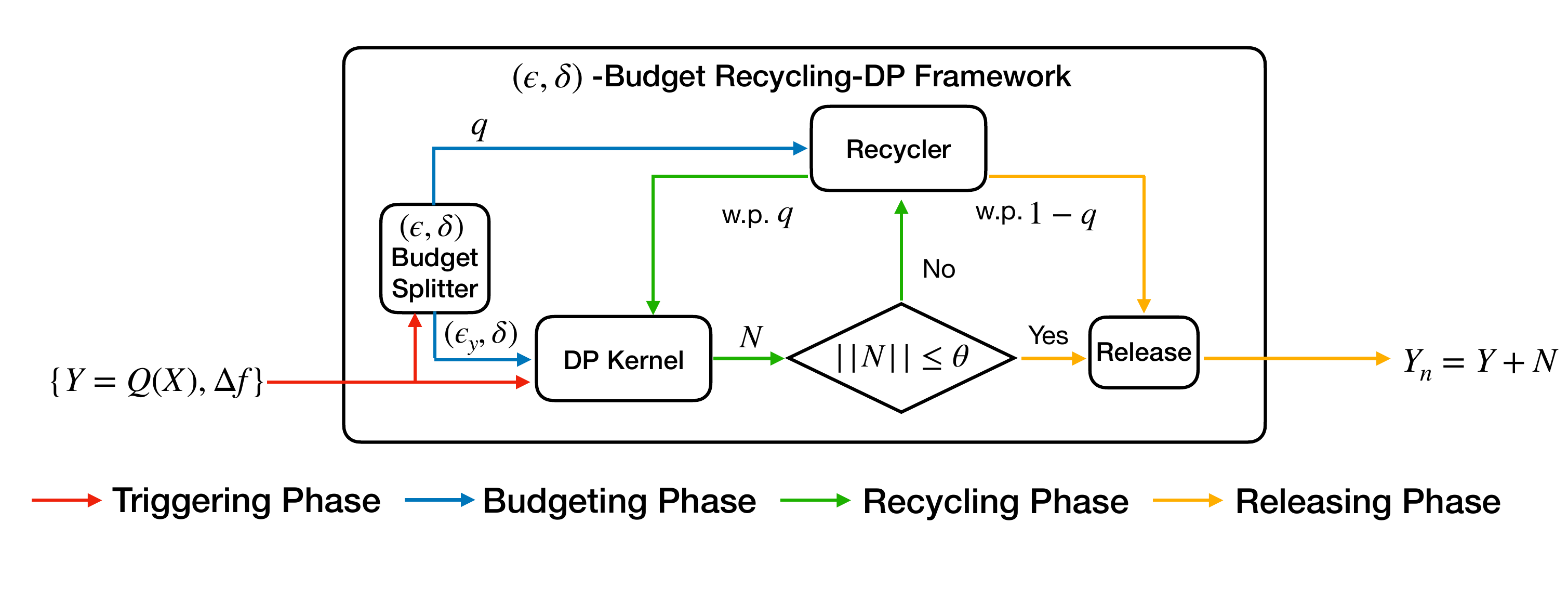}
    \caption{Illustration of the Budget Recycling Differential Privacy (BR-DP) framework}
    \label{fig:BR-DP_model}
\end{figure*}

There are many attempts trying to improve the \textit{acceptance} rate provided by DP, the first direction is adopting bounded noisy distribution \cite{Geng2018TruncatedLM, 2022-2-7-0193, holohan2018bounded}. This technique is usually achieved by re-sampling or truncating. However, most of the current solutions assume the support of the output is bounded and fixed, and hence do not directly measure the noise magnitude. In \cite{Geng2018TruncatedLM}, Geng et al. proposed a bounded Laplacian mechanism in which the noise magnitude is bounded. However, the proposed mechanism requires a large failure probability to handle the output support discrepancy caused by introducing sensitivity. This drawback is even amplified when considering composition. There's also a line of research that focuses on improving the utility measured by an absolute error by considering specific data distribution, query types, and mechanism consistency \cite{10.1145/1806689.1806786, 10.14778/1920841.1920970, 10.1145/1807085.1807104}. Works that relax the DP guarantee such as Geo-indistinguishably\cite{10.1145/2508859.2516735}, Membership Privacy\cite{10.1145/2508859.2516686}, Pufferfish Privacy\cite{Kifer2014PufferfishAF,Pufferfish_mec}, Information Privacy\cite{LIP1} are effective, but are not in the same pool here, since they are not providing rigorous DP guarantee. 

\section{Budget Recycling Differential Privacy}
In this section, we introduce the Budget Recycling Differential Privacy (BR-DP) framework. We begin with an overview of the components in this framework, followed by a detailed explanation of its noisy distribution.

\subsection{Mechanism Overview}
The structure of our proposed BR-DP framework is illustrated in Fig. \ref{fig:BR-DP_model}. At its core, the mechanism functions as a traditional differential privacy (DP) mechanism, termed the DP kernel, augmented by a Budget Recycling module (Recycler). This module either releases a noisy answer or recycles the budget for regeneration with a certain randomness. The framework comprises four distinct phases:

\textbf{Triggering Phase:} Initiated by a query, the BR-DP framework activates with the initialization of the following parameters: $(\epsilon, \delta)$, representing the total budget; $\Delta_f$, denoting the estimated query sensitivity; the type of noisy mechanism for the DP kernel, such as Gaussian or Laplacian; and $\theta$, defining the error boundaries. 

\textbf{Budgeting Phase:} Adhering to an optimal budget allocation principle (discussed in Section \ref{sec:budgeting}), the total budget is strategically divided. This division comprises $(\epsilon_y, \delta_y)$ allocated to the DP kernel and a recycling probability $q$ designated for the recycler, collectively satisfying the $(\epsilon, \delta)$-DP requirement.

\textbf{Releasing Phase:} The DP kernel, equipped with $(\epsilon_y, \delta_y)$ and $\Delta_f$, generates a noise distribution from which noise $n$ is sampled. This sample is in accordance with the stipulated utility requirement. If $||n|| \leq \theta$, the noise $n$ is appended to the true answer $y$, resulting in the release of $y_n = y + n$. Conversely, if $||n|| > \theta$, there is a $(1 - q)$ probability that $n$ is appended and released.

\textbf{Recycling Phase:} In instances where $||n|| > \theta$, there exists a $q$ probability that both $n$ and $\epsilon_y$ are recycled. Subsequently, the mechanism reverts to the releasing phase to regenerate $Y_n$. This iterative process continues until a noisy version of $y_n$ is released.

From the descriptions above, we observe that BR-DP connects between a conventional DP mechanism and the bounded DP mechanism through a recycling parameter $q$: $q = 1$ implies a bounded DP mechanism and $q=0$ means a conventional DP mechanism. Therefore, BR-DP can be regarded as a soft-bounded mechanism. This work seeks to optimize this $q$ to achieve the optimal privacy-utility tradeoff, specifically, Adopting the BR-DP framework brings the following advantages:

\begin{enumerate}
    \item \textbf{Differential Privacy Protection:} The BR-DP framework ensures differential privacy guarantee with any arbitrary but fixed privacy budget $(\epsilon, \delta)$. Note that it is infeasible for some other bounded DP mechanisms such as \cite{Geng2018TruncatedLM} to achieve {$(\epsilon, \delta)$}-DP under small $\delta$ or large query sensitivity;

    \item \textbf{Enhanced Utility:} The BR-DP framework significantly improves data utility. This is evidenced by the probability $||N|| \leq \theta$ for any predetermined threshold $\theta$.

    \item \textbf{Tightened Composition:} The BR-DP framework offers a tighter composition compared to a standard DP kernel mechanism assigned the same privacy budget. Specifically, a series of releases from sequential $(\epsilon, \delta)$-BR-DP frameworks results in less privacy leakage than an equivalent series of sequential $(\epsilon, \delta)$-DP mechanisms.

    \item \textbf{Versatility with Other DP Mechanisms:} The BR-DP framework is highly versatile and can enhance any DP mechanism. It is not restricted to a specific type of noisy distribution, thus broadening its applicability in various differential privacy scenarios. 

\end{enumerate}

\begin{figure*}[t]
\centering 
\subfigure[Illustration of the noisy distribution of the BR-DP framework with Gaussian kernel.]
{\includegraphics[width=0.55\textwidth]{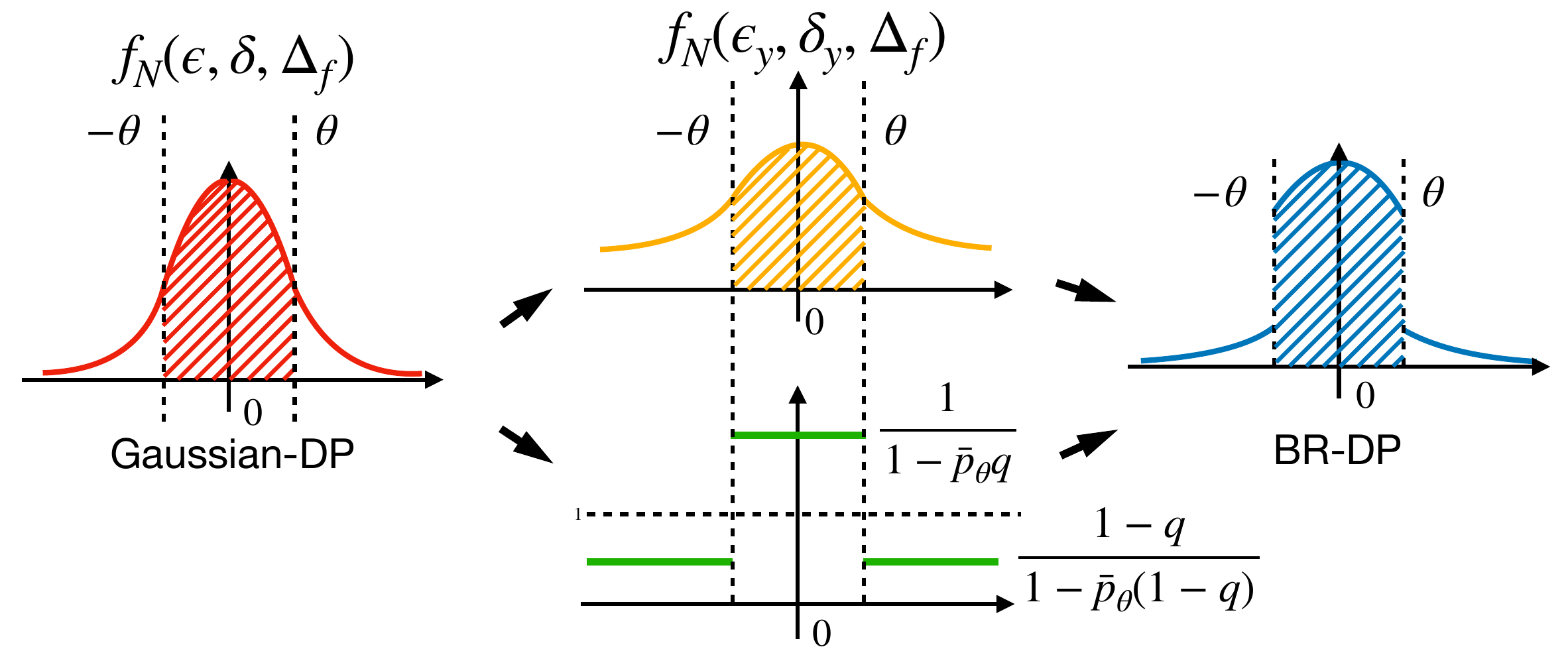}
\label{fig:2-1}}~
\subfigure[Gaussian kernel BR-DP v.s. Gaussian DP for data generation. In each mechanism, $\epsilon = 2$, $\delta = 10^{-5}$, $\theta=1$.]
{\includegraphics[width=0.37\textwidth]{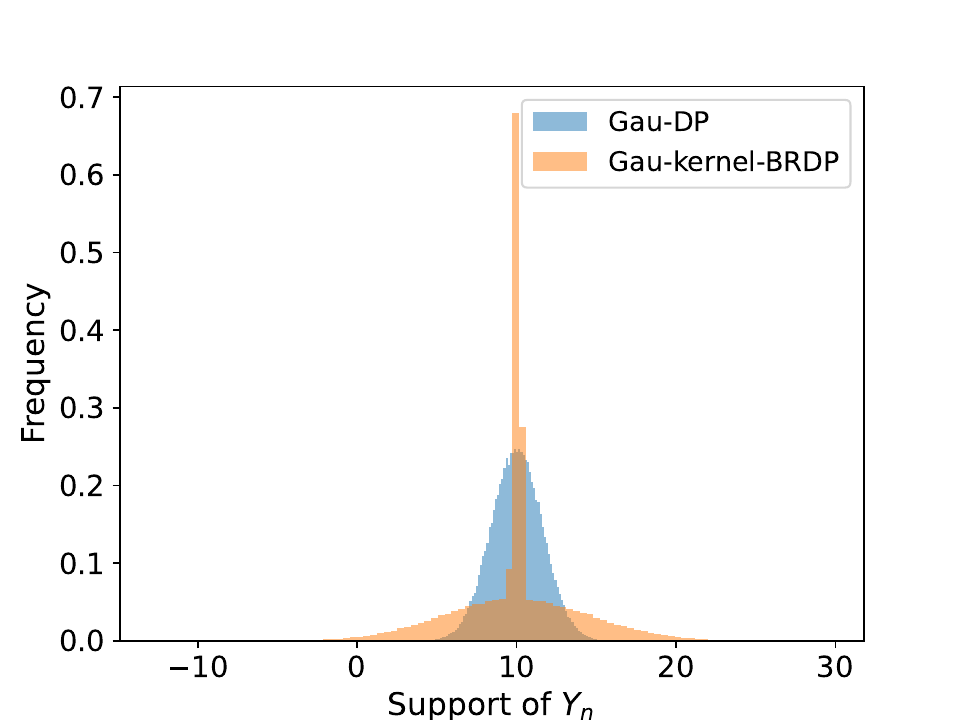}
\label{fig:2-2}}
\label{fig:2}
\caption{BR-DP noise distribution illustration}
\end{figure*}

\subsection{Noisy distribution of BR-DP}
The noisy distribution in the BR-DP framework is influenced by three factors: the noise distribution of the kernel DP mechanism, the utility requirement specified by $||N|| \leq \theta$, and the recycling rate $q$. We represent the noise probability density function (pdf) of the kernel DP mechanism as $f_N$, and a conditioned noisy pdf $f_N(\cdot|y)$ denotes $f_N$ centered at $Y=y$.  Denote $S_{\theta}$ as the subset of support of $Y_n$ corresponding to $||N||\le{\theta}$ given $y$. The notation ${p}_{\theta}$ denotes the probability $\text{Pr}(Y_n \in S_{\theta})$ obtained with the distribution of $f_N$, and $\bar{p}_{\theta}$ represents the complement. 
\begin{prop}
The noisy distribution of BR-DP centered at $y$ can be derived as
\begin{equation}\label{noise_dist}
    f_{\text{brdp}}(y_n|y) =
    \begin{dcases} 
        \frac{f_N(y_n|y)}{1-\bar{p}_{\theta}q}, & \text{if } y_n\in S_{\theta}; \\
        \frac{f_N(y_n|y)(1-q)}{1-\bar{p}_{\theta}q}, & \text{otherwise}. 
    \end{dcases}
\end{equation}
\end{prop}
The steps for deriving the BR-DP distribution are summarized in the Appendix.
From \eqref{noise_dist}, we note the following observations:
\begin{enumerate}
    \item  $f_{brdp}$ depends on $f_N$, $\bar{p}_{\theta}$ and $q$, for large $q\to 1$, $f_{brdp}$ is a normalized $f_N$ bounded by $\mathcal{S}_{\theta}$; for small $q \to 0$, $f_{brdp}$ becomes identical to $f_N$. 
    \item The mechanism amplifies the probability of releasing a $y_n$ if $y_n \in S_{\theta}$; The amplification is of a factor of $1/({1-\bar{p}_{\theta}q})$, i.e., the amplified probability enlarges with $q$ and $\bar{p}_{\theta}$.
    \item The mechanism decreases the probability of releasing a $y_n$ that $y_n\notin  S_{\theta}$, the decrease is of a factor of  ${(1-q)}/{(1-\bar{p}_{\theta}q)}$.
\end{enumerate}

The BR-DP noisy distribution can be considered a soft-bounded distribution. While the bounded support of $f_N$ ensures utility, the flexible rules for recycling or releasing $Y_n$ beyond these boundaries introduce ambiguity and enhance privacy. Figure \ref{fig:2-1} depicts the noisy distribution of $Y_n$, and in Figure \ref{fig:2-2}, we present the normalized frequency of $100,000$ data points generated using the Gaussian kernel BR-DP and the Gaussian DP mechanism. It is noteworthy that BR-DP ensures that most of the generated $y_n$ fall within $S_{\theta}$, albeit at the expense of increased variance. This expansion in the support of the BR-DP noise, compared to standard DP, leads to less precise unacceptable noisy generations, thereby reinforcing DP protection overall. Moreover, as discussed in Section \ref{sec:Composition}, this increased variance results in reduced leakage when composing multiple mechanisms.

\section{Privacy Leakage Analysis}
In this section, we delve into the privacy leakage associated with the BR-DP framework by scrutinizing its privacy loss distribution. Then we derive optimal $q$ under any $(\epsilon_y,\delta_y)$ for the DP kernel and $(\epsilon,\delta)$ as the total budget.

\subsection{Privacy Leakage Accounting for BR-DP}
There are many interpretations for DP, such as KL divergence, total variance, Renyi divergence, hypotheses testing, etc. However, all these options for quantifying indistinguishably can be viewed from the perspective of Privacy Loss Distribution (PLD). We next formally define PLD.
\begin{defi}[Privacy Loss Distribution] Let $P$ and $Q$ be two probability distributions on $\mathcal{Y}$. Define $f_{P/Q}: \mathcal{Y}\to \mathsf{R}$ by $f_{P/Q}(y_n) = \log(P(y_n)/Q(y_n))$. The privacy loss random variable is given by $Z = f_{P/Q}(Y_n)$. The distribution of $Z$ is denoted $f_Z(z)$.
\end{defi}
In the following of this paper, we use $Z$ to denote the privacy loss random variable of the kernel DP mechanism of BR-DP. Specifically, for each $y_n\in\mathcal{Y}$, $Z$ takes values of:
\begin{equation*}
    \log\left(\frac{\mathcal{M}(X)=y_n}{\mathcal{M}(X')=y_n}\right)=\log\left(\frac{f_N(y_n|y)}{f_N(y_n|y+\Delta_f)}\right),
\end{equation*}
with a probability of $f_N(y_n|y)$.

Given $Z$, the privacy profile is readily obtained by: 
\begin{equation*}
\begin{aligned}
    \delta \ge &\mathbb{E}_{Z}[\max\{0,1-\exp(\epsilon-z)\}]\\
    =&\int_{\epsilon}^{\infty}(1-\exp(\epsilon-z))f_Z(z)dz.\\
\end{aligned}
\end{equation*}
For simplicity, we use $\delta_Z(\epsilon)$ to denote $\delta$ for a given $\epsilon$, and the PLD of $Z$. Let $\tau_u$ be the positive upper-utility bound, and $\tau_l$ be the negative upper-utility bound. Then,  the probability of $p_{\theta}$ can be represented by the Cumulative Density Function (CDF) of $N$, which is denoted as $\Phi(\cdot)$. Since  $\Phi(\cdot)$ depends on the variance (we assume zero-mean for kernel DP mechanism), which further depends on $\epsilon_y$ and $\delta_y$. Here, we denote $\Phi_{(\epsilon_y,\delta_y)}(\cdot)$ as the CDF of a noisy distribution given $(\epsilon_y,\delta_y)$. Then:
\begin{equation}\label{eq:pytheta}
    p_{\theta}=\Phi_{(\epsilon_y,\delta_y)}(\tau_{u})-\Phi_{(\epsilon_y,\delta_y)}(\tau_{l}),
\end{equation}
\begin{equation}\label{eq:barpytheta}
    \bar{p}_{\theta}=1-\Phi_{(\epsilon_y,\delta_y)}(\tau_{u})+\Phi_{(\epsilon_y,\delta_y)}(\tau_{l}).
\end{equation}

Denote $\Gamma$ as the privacy loss random variable of the BR-DP framework, we next introduce the PLD of BR-DP represented as a parameter of $\Phi$.

\begin{thm}\label{PLD_BR-DP}
The PLD of a BR-DP framework, given the PLD of the kernel DP mechanism $f_Z$ and a recycling rate $q$, can be represented as:
\begin{equation*}
    f_{\Gamma}(\gamma)=(1-W) f_Z(\gamma)+W f_Z(\gamma-\mathcal{L}),
\end{equation*}
where 
\begin{equation}\label{eq:l}
    \begin{aligned}
    &\mathcal{L} = -\log\left({1-q}\right),\\
\end{aligned}
\end{equation}
\begin{equation}\label{eq:W}
    \begin{aligned}
    &W = \max\{(\Phi_{(\epsilon_y,\delta_y)}(\min\{\tau_l + \Delta_f,\tau_u\})-\Phi_{(\epsilon_y,\delta_y)}(\tau_l)),\\&~~~(\Phi_{(\epsilon_y,\delta_y)}(\tau_u + \Delta_f)-\Phi_{(\epsilon_y,\delta_y)}(\max\{\tau_u, \tau_l + \Delta_f\})\}.
    \end{aligned}
\end{equation}
\end{thm}
Intuitively, the PLD of the BR-DP framework shifts the PLD of the kernel DP mechanism to the right (indicating a larger privacy leakage) by \( \mathcal{L} \) with a probability of \( W \). Consequently, minimizing either \( \mathcal{L} \) or \( W \) can contribute to reducing the overall privacy leakage. A smaller \( \mathcal{L} \) or a lower \( q \) serves to limit the increase in additional leakage, while a reduced \( W \) decreases the likelihood of such leakages occurring. We next introduce the privacy profile of BR-DP.

\begin{prop}
The privacy profile of the BR-DP framework can be derived as follows:
\begin{equation}\label{eq:deltaGamma}
    \delta_{\Gamma}(\epsilon) \ge (1-W)\delta_Z(\epsilon)+W\delta_Z(\epsilon-\mathcal{L}).
\end{equation}
\end{prop}

The privacy profile of BR-DP depends on the privacy profile of the kernel DP mechanism, which has been thoroughly studied in the literature as discussed in the preliminary DP section. Specifically, the privacy profile of BR-DP can be derived by a linear combination of two privacy profiles of the kernel DP mechanism with different $\epsilon$'s. 

\subsection{Determining Recycling Rate}

\textbf{Baseline $q$:} We first propose naive analysis on the leakage accounting by assuming $\delta_y=\delta$. Specifically, if the BR-DP kernel is $(\epsilon_y, \delta)$ differentially private, the leakage upper bound of the BR-DP framework follows the next Theorem.

\begin{thm}\label{thm:leak_bound}
The BR-DP framework with $(\epsilon_y,\delta)$- kernel DP mechanism and $q$ as the recycling rate for the recycling module satisfies $(\epsilon,\delta)$-DP, where 
\begin{equation*}
    \epsilon = \epsilon_y - \log(1-q).
\end{equation*}
\end{thm}
Then, it is straightforward to derive the baseline $q$ achieving $(\epsilon,\delta)$-DP given a $(\epsilon_y,\delta)$-kernel DP.
\begin{rmk}\label{rmk1}
The baseline $q$ that guarantees $(\epsilon,\delta)$-DP is $$1-e^{-(\epsilon-\epsilon_y)}.$$
\end{rmk}
Here, baseline means that $q$ is derived from the naive leakage accounting in Theorem \ref{thm:leak_bound}, and the true leakage for a BR-DP framework by adopting $q$ for the recycler and $(\epsilon_y,\delta)$ for the DP kernel, is less than $(\epsilon,\delta)$. 

\textbf{Optimal $q$:} Observe from \eqref{eq:deltaGamma}$, \delta_{\Gamma}$ depends on $\delta_Z$, $\mathcal{L}$, and $W$, which further depends on kernel DP parameters, and the recycling rate $q$. To this end, $q$ can be obtained when the total budget and the kernel DP budget are given, by employing by a binary search algorithm summarized in Algorithm \ref{alog:findq}. 

\begin{algorithm}\label{alog:findq}
\caption{Find $q$ with binary search}
\hspace*{\algorithmicindent}
\textbf{Input:} $\epsilon_y, \delta_y$, $\epsilon, \delta$, $\Delta_f$ $\theta$, tol.\\
 \hspace*{\algorithmicindent} \textbf{Output:} Optimal $q$.
\begin{algorithmic}[1]
    \State Given $\epsilon_y$, $\delta_y$, $\Delta_f$, get $\Phi_{(\epsilon_y,\delta_y)}$ and $\delta_Z(\cdot)$;
    \State  Get $\tau_l$ and $\tau_u$ with $\theta$;
    \State Initialize $q_{up} \gets 1$, $q_{low} \gets 0$;
    \State $p_{\theta}\gets$ \eqref{eq:pytheta}, $\bar{p}_{\theta}\gets$ \eqref{eq:barpytheta};
    \State $W \gets$ \eqref{eq:W};
    \While{$q_{up} - q_{low} > \text{tol}$}
        \State $q_{mid} \gets (q_{up} + q_{low})/2$;
        \State $\mathcal{L}$ $\gets$ \eqref{eq:l};
        \State $\delta_{\Gamma}(\epsilon) \gets$ \eqref{eq:deltaGamma};
        \If{$\delta_{\Gamma}(\epsilon) \le \delta_{target}$}
            \State $q_{low} \gets q_{mid}$;
        \Else
            \State $q_{up} \gets q_{mid}$;
        \EndIf
    \EndWhile
    \State \Return $q_{low}$
\end{algorithmic}
\end{algorithm}

\begin{figure}
    \centering
\includegraphics[width=0.43\textwidth]{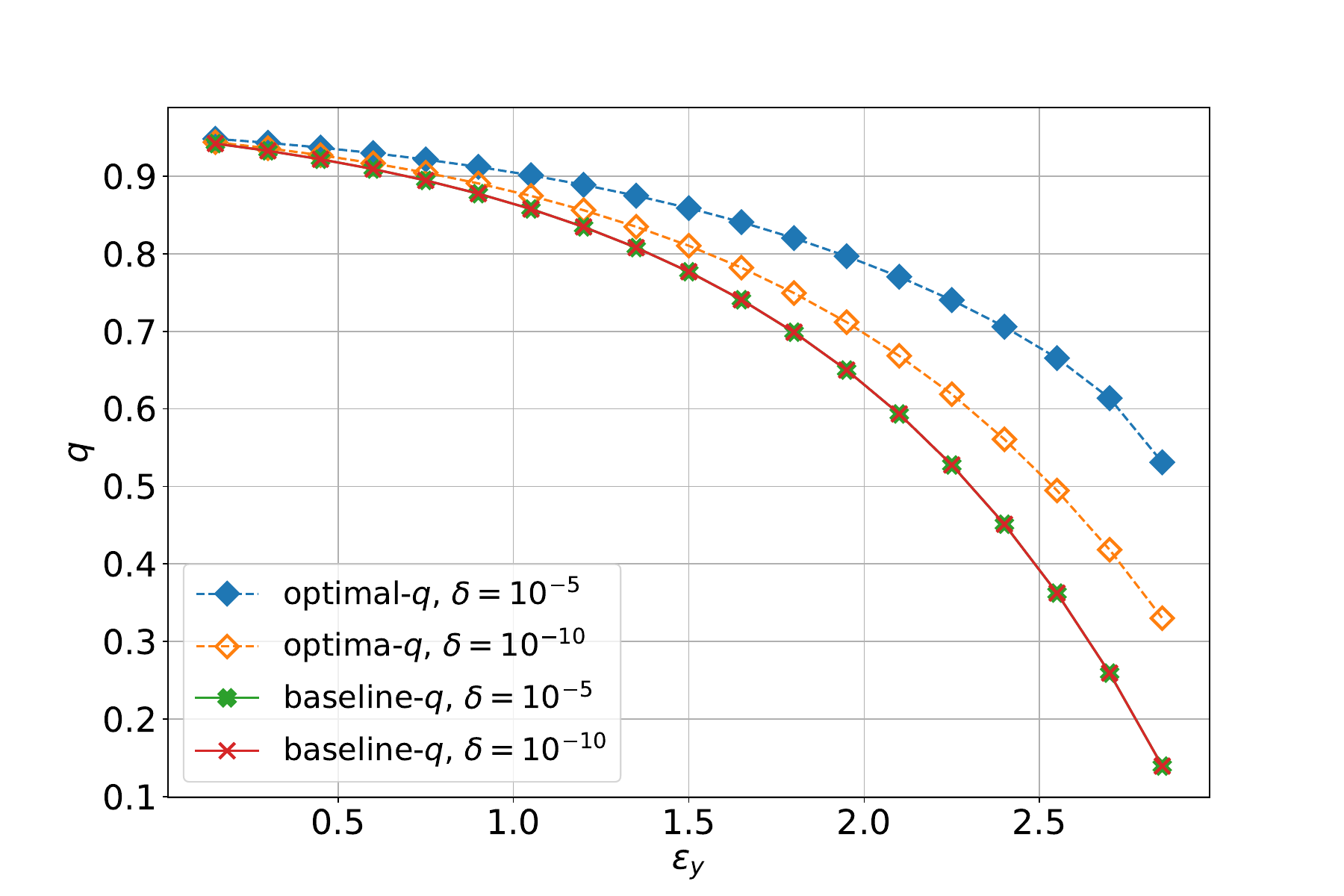}
    \caption{Comparison of $q$ for baseline and optimal approaches. $\epsilon = 3$, $\delta = \delta_y$ with values shown in the figure, $y=10$, $\Delta_f = 1$, $\theta = 1$.}
    \label{fig:q_comp}
    \vspace{-10pt}
\end{figure}


{It is worth emphasizing here that $q$ is only determined by the total privacy budget, budget allocated to the kernel DP, query sensitivity, error bounds and tolerance factor in numerical searching. Hence, $q$ is data independent.}
The following proposition states the privacy guarantee.
\begin{prop}
The BR-DP framework incorporates an \((\epsilon_y, \delta_y)\)-DP kernel and employs a recycling rate \( q \) as determined by Algorithm 1 for its recycling module, is shown to satisfy \((\epsilon, \delta)\)-DP.
\end{prop}


We conduct a numerical analysis 
to compare the trade-off between \( q \) and \( \epsilon_y \) for a given $\epsilon$, with $q$
derived from Algorithm 1 and the baseline analysis in Remark \ref{rmk1}, respectively. This analysis is based on a Gaussian kernel BR-DP, with results
depicted in Fig. \ref{fig:q_comp}. Observe that Algorithm 1 brings much-tightened results by providing larger $q$.

\section{Utility-Oriented Optimization}

In this section, we optimize the allocation of the budget $(\epsilon, \delta)$ within the BR-DP framework. Our goal is to optimally allocate the budget for the kernel DP mechanism in a way that maximizes the utility achievable by BR-DP. To this end, we first explore the utility function of the BR-DP framework with the BR-DP noisy distribution, based on which, we propose the principles of optimal budgeting for the kernel DP mechanism.


Given the noisy distribution of BR-DP defined in \eqref{noise_dist}, the utility, measured by the \textit{acceptance rate} becomes:
\begin{equation}\label{eq:post}
\begin{aligned}
    \text{Pr}(||N||\le{\theta})=&\int_{S_{\theta}}f_{brdp}(y_n|y)dy_n
    =\frac{p_{\theta}}{1-\bar{p}_{\theta}q}
\end{aligned}
\end{equation}


Observe that, to maximize utility, the framework faces a dilemma: either to enhancing \( p_{\theta} \), or increasing \( q \). Specifically, a larger budget allocated to the kernel DP mechanism leads to a noise distribution with a small variance that tends to be more concentrated around the ground truth, and hence, improves \( p_{\theta} \), whereas increasing \( q \) necessitates diverting more budget to the recycling module. Consequently, the optimal BR-DP framework is characterized by a strategic and balanced budget distribution between the kernel DP mechanism and the recycling module.

 \begin{algorithm}[t]
\caption{Find Optimal $\epsilon_y$ using Ternary Search}
\hspace*{\algorithmicindent}
\textbf{Input:} $\epsilon$, $\delta$, $\Delta_f$, $\theta$, tol.\\
 \hspace*{\algorithmicindent} \textbf{Output:} Optimal  $\epsilon_y$.
\begin{algorithmic}[1]
    \State $\epsilon_{low} \gets 0$;
    \State $\epsilon_{up} \gets \epsilon$;
    \While{$\epsilon_{up} - \epsilon_{low} > \text{tol}$}
        \State $\epsilon_1 \gets \epsilon_{low} + \frac{\epsilon_{up} - \epsilon_{low}}{3}$;
        \State $\epsilon_2 \gets \epsilon_{up} - \frac{\epsilon_{up} - \epsilon_{low}}{3}$;
        \State $q_1 \gets$ Algorithm 1($\epsilon, \delta, \epsilon_1, \delta, \Delta_f,\theta$, tol);
        \State $q_2 \gets$ Algorithm 1($\epsilon, \delta, \epsilon_2, \delta, \Delta_f,\theta$, tol);
        \If{{O}{($\epsilon_1, q_1$)} $>${O}{($\epsilon_2, q_2$)}}
            \State $\epsilon_{low} \gets \epsilon_1$;
        \Else
            \State $\epsilon_{up} \gets \epsilon_2$;
        \EndIf
    \EndWhile
    \State $\epsilon_y \gets (\epsilon_{up} + \epsilon_{low})/ 2$;
    \State \Return $\epsilon_y$
\end{algorithmic}
\label{algo:2}
\end{algorithm}

\subsection{Optimal Budgeting Principles}\label{sec:budgeting}
We now proceed to derive the principle for optimal budget allocation within the BR-DP framework, given a total budget of $(\epsilon, \delta)$. For simplicity in our derivation, we allocate the same $\delta$ to the kernel DP mechanism as in the total budget. However, our results can be readily extended to accommodate an arbitrary $\delta_y$ for the kernel mechanism.

For a given $\epsilon_y, \delta$, the \textit{acceptance rate} in \eqref{eq:post} can be further expressed as:

\begin{equation*}
    \begin{aligned}
\frac{p_{\theta}}{1-\bar{p}_{\theta}q}
=&\frac{\Phi_{(\epsilon_y,\delta)}(\tau_{u})-\Phi_{(\epsilon_y,\delta)}(\tau_{l})}{[\Phi_{(\epsilon_y, \delta)}(\tau_{u})-\Phi_{(\epsilon_y,\delta)}(\tau_{l})](1-q)+q}\\
=&\left({1-q+\frac{q}{\Phi_{(\epsilon_y,\delta)}(\tau_{u})-\Phi_{(\epsilon_y,\delta)}(\tau_{l})}}\right)^{-1}.
\end{aligned}
\end{equation*}
Equivalently, to improve the \textit{acceptance rate}, the mechanism tends to minimize
\begin{equation}\label{eq:budget}
O(\epsilon_y,q) = 1-q+\frac{q}{\Phi_{(\epsilon_y,\delta)}(\tau_{u})-\Phi_{(\epsilon_y,\delta)}(\tau_{l})}.
\end{equation}

\begin{figure}[t]
\centering 
\subfigure[\textit{Acceptance rate} v.s. the total privacy budget $\epsilon$, when $\theta = 1$, $\Delta_f = 1$, $\delta_y = \delta$.]
{\includegraphics[width=0.4\textwidth]{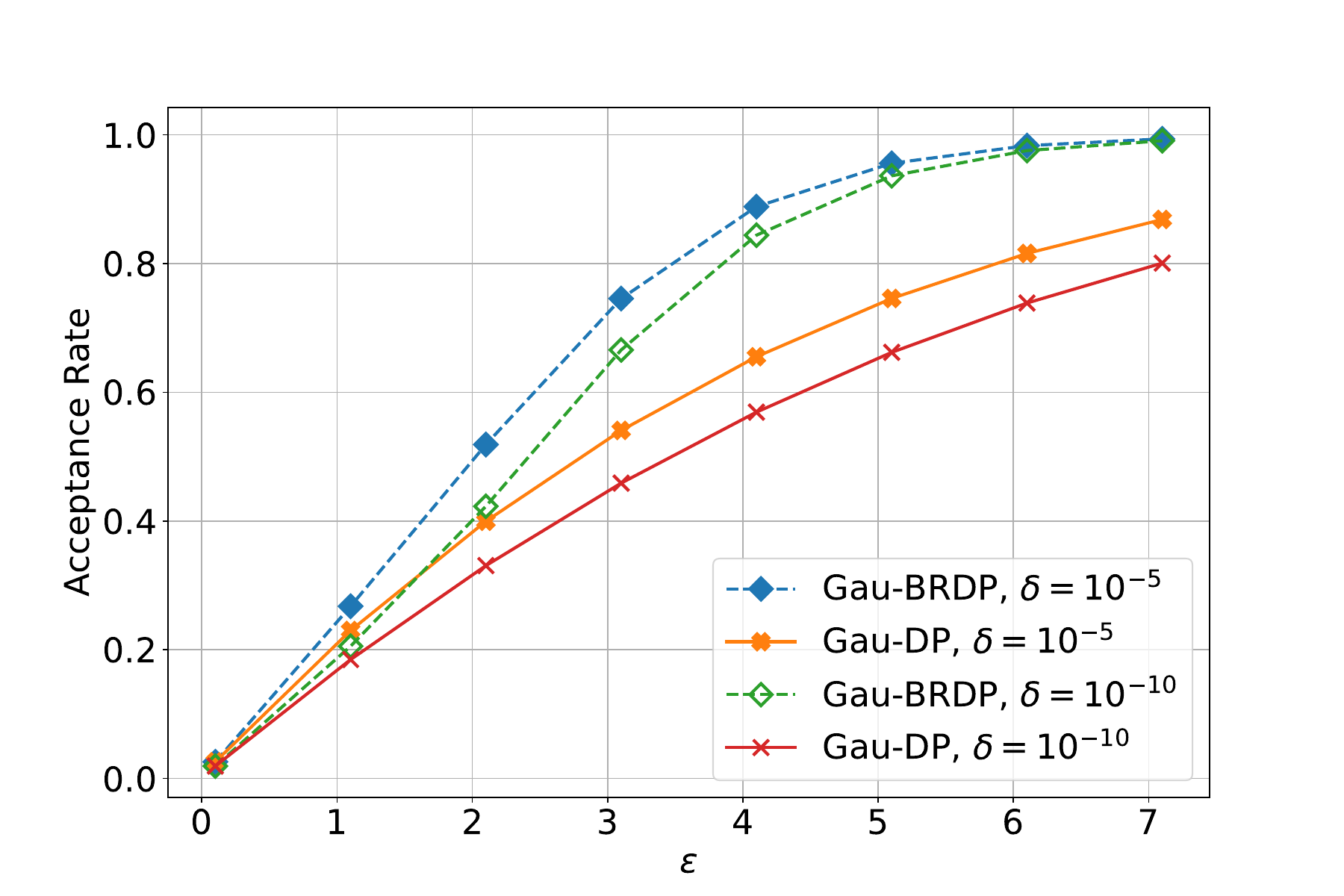}
\label{ar_vs_ep}}
\subfigure[Acceptance rate v.s.  the total privacy budget $\epsilon$, when $\theta = 1$, $\Delta_f = 5$, $\delta_y = \delta$.]
{\includegraphics[width=0.4\textwidth]{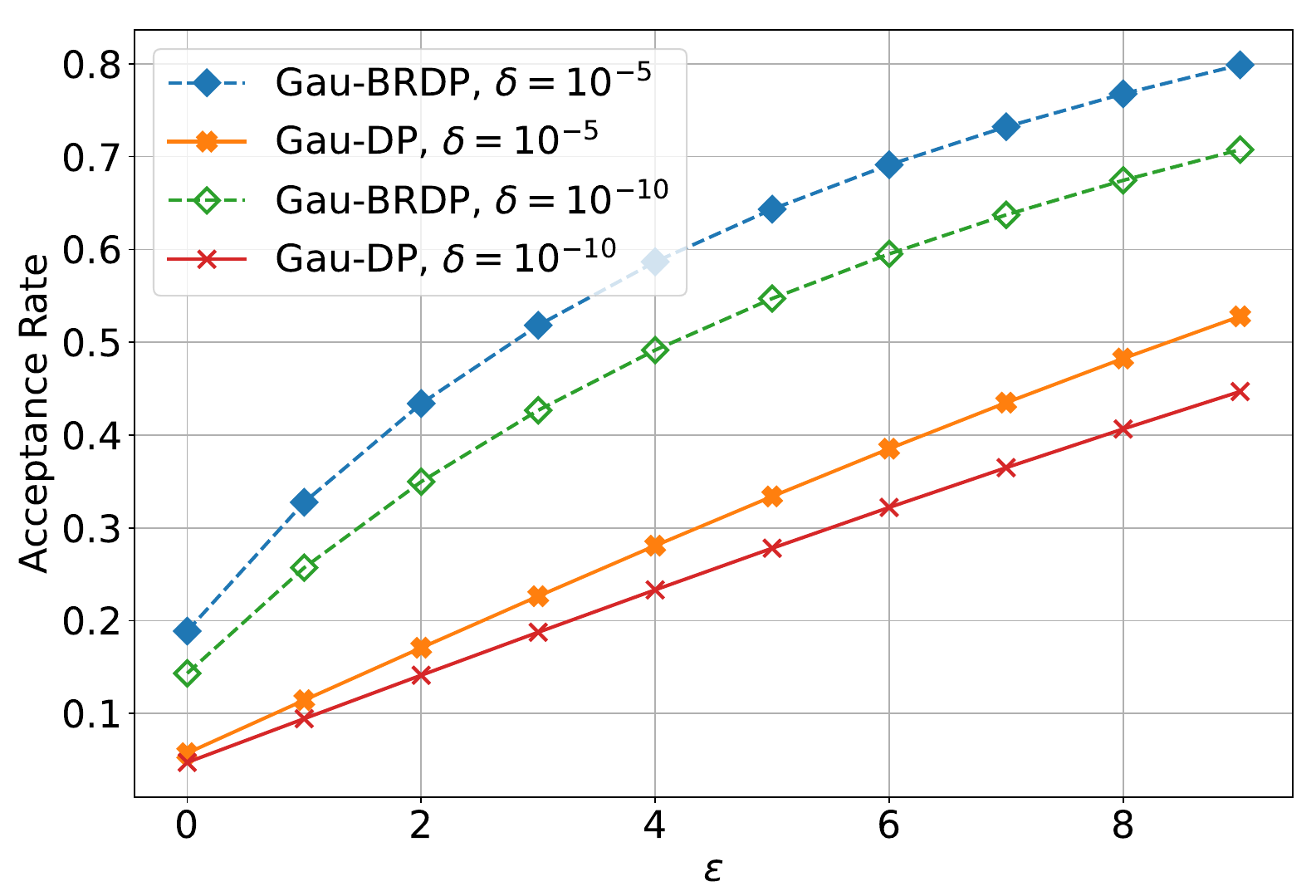}
\label{ar_vs_y1}}
\caption{Numerical comparison of the acceptance rate provided by Gaussian DP and the Gaussian kernel BR-DP under different parameters.}
\vspace{-10pt}
\end{figure}

Given that the value of \( q \) can be determined from \( \epsilon_y \) as outlined in Algorithm 1, Equation \eqref{eq:budget} can be treated as a constrained one-dimensional optimization problem as a function of $\epsilon_y$. Several numerical methods are applicable for solving this type of problem. As an illustrative example, we employ the Ternary Search algorithm as outlined in Algorithm \ref{algo:2}. Ternary search excels in efficiently pinpointing the optimal solution for objective functions that display a single peak across the support of the input variable. Crucially, the function should be monotonic both before and after this peak. In our context, the objective function is parameterized by $\epsilon_y$. When $\epsilon_y$ is large, the dominance of the DP kernel causes the utility to be driven by it, attenuating the recycling module. Conversely, with a small $\epsilon_y$, the recycling phase takes precedence, thereby governing the utility. This behavior of the objective function mirrors the characteristics suited for the Ternary search.

Comparing the utility provided by BR-DP and DP, we have the following conclusion.
\begin{prop}
The $(\epsilon,\delta)$-BR-DP framework achieves an equal or superior \textit{acceptance rate} compared to the 
$(\epsilon,\delta)$-DP mechanism when the kernel DP mechanism within BR-DP utilizes the same form of noisy distribution as that employed in the DP mechanism.
\end{prop}
\begin{proof}
 Observe that when $q=0$, the $O(\epsilon_y,q)$ becomes:
\begin{equation*}
{(\Phi_{(\epsilon,\delta)}(\tau_{u})-\Phi_{(\epsilon,\delta)}(\tau_{l}))^{-1}}.
\end{equation*}   
Given that $q=0$, the parameter $\epsilon_y = \epsilon$ falls within the searchable region as defined by Algorithm 2. This implies that the worst-case performance of BR-DP occurs when it reduces to its kernel DP. 
\end{proof}

We conduct a numerical comparison of the \textit{acceptance rates} between Gaussian DP and BR-DP with a Gaussian kernel. This involves illustrating the utility-privacy tradeoff by varying $\epsilon$ and plotting the respective \textit{acceptance rates}. We explore two scenarios based on the sensitivity: case (a) with $\Delta_f = 1$ and case (b) with $\Delta_f = 5$. For BR-DP, we determine the parameters using the optimal $\epsilon_y$  allocation algorithm (Algorithm 2) and the optimal $q$ determination algorithm (Algorithm 1). It is notable that BR-DP consistently achieves higher \textit{acceptance rates} compared to DP, with this advantage being more pronounced for queries with higher sensitivity.



\section{Composabilities}\label{sec:Composition}
In the context of differential privacy, composition refers to the property that allows the combination of multiple differentially private mechanisms, each providing a certain level of privacy, and provides a way to calculate the overall level of privacy offered when these mechanisms are applied together. There are two main types of composition in differential privacy: parallel composition and sequential composition. In this section, we discuss composabilities of BR-DP, we first address basic parallel and sequential composition results. Then, focusing on the sequential composition, we provide a tight analysis.
\subsection{Basic Composition}
As a privacy-preserving mechanism achieving DP, the BR-DP framework satisfies most of the basic properties of the DP. Such as post-processing, and linkage properties. etc. Specifically, the following remark presents the parallel composability of BR-DP.
\begin{rmk}
    For $T$ independent BR-DP mechanisms $\mathcal{M}_1,...,\mathcal{M}_T$, each providing $(\epsilon,\delta)$-DP, and each mechanism is applied to a disjoint subset of the overall dataset, then $(\mathcal{M}_1,...,\mathcal{M}_T)$ still provides $(\epsilon,\delta)$-differential privacy.
\end{rmk}
The sequential composability, on the other hand, examines the leakage for a series of independent mechanisms applied to the same dataset.

\begin{rmk}
For $T$ independent BR-DP mechanisms $\mathcal{M}_1,...,\mathcal{M}_T$, each providing $(\epsilon,\delta)$-DP, and each mechanism is applied sequentially to the same dataset, then $(\mathcal{M}_1,...,\mathcal{M}_T)$ provides $(T\epsilon,T\delta)$-DP.
\end{rmk}

Remark 3 indicates that the composed leakage for multiple BR-DP mechanisms grows linearly with the number of mechanisms applied. The following remark shows that the advanced composition for DP still applies to BR-DP:

\begin{rmk}
    The sequence of mechanisms $\mathcal{M}(1:T)$ in Remark 3 satisfies $(T\epsilon(e^{\epsilon}-1)+\sqrt{T}\epsilon\sqrt{2\ln(1/\delta)},\delta)$-DP.
\end{rmk}

\subsection{Tight Composition Analysis}\label{sec:tightcompo}

We next derive more tightened results on the sequential composability of BR-DP. From the PLD of BR-DP, we let
\begin{equation*}
    f_{\Gamma}=f_Z\ast f_R,
\end{equation*}
where $f_R(r)$ denotes the privacy loss distribution caused by the recycling phase, and 
\begin{equation*}
    f_{R}(r) = (1-W)\delta_{\text{Dirac}}(r)+W\delta_{\text{Dirac}}(r-\mathcal{L}),
\end{equation*}
and $\delta_{\text{Dirac}}$ represent the Dirac function such that $\delta_{\text{Dirac}}(r) = 1$, iff $r =0$, otherwise   $\delta_{\text{Dirac}}(r) = 0$. Then, the following proposition describes the PLD of a BR-DP framework after $T$-fold compositions:


\begin{prop}
    Consider a BR-DP framework characterized by a Privacy Loss Distribution (PLD) \( f_{\Gamma}(\gamma) \). After undergoing a \( T \)-fold non-adaptive composition, the resulting PLD  can be expressed as:
    \begin{equation*}
        f^T_{\Gamma}(\gamma) = (f_{\Gamma} \ast^T f_{\Gamma})(\gamma),
    \end{equation*}
    where \( \ast^T \) denotes the \( T \)-fold convolution operation of the PLD. Furthermore, the closed-form expression of this PLD is given by:
    \begin{equation}
        \sum_{k=0}^T \binom{T}{k} (1-W)^k W^{T-k} (f_Z \ast^T f_Z) (z-(T-k)\mathcal{L}),
    \end{equation}
    with \( \mathcal{L} \) and \( W \) being parameters defined previously in \eqref{eq:l} and \eqref{eq:W}, respectively.
\end{prop}

Then, regarding the privacy profile of BR-DP after composition, we have the following Theorem.
\begin{thm}\label{thm:compo}
    Consider a T-fold non-adaptive composition of a BR-DP framework. The composition is tightly $(\epsilon,\delta)$-DP for $\delta(\epsilon)$ given by
    \begin{equation*}
        \delta_{\Gamma}^T(\epsilon)=\sum_{k=0}^T\binom{T}{k} (1-W)^{k}W^{T-k}\delta^T_{Z}(\epsilon-(T-k)\mathcal{L}).
    \end{equation*}
   where $\delta_Z^T(z)$ denotes the privacy profile of the kernel DP mechanism after $T$-fold composition. 
\end{thm}
It is important to emphasize that the sequential composition analysis can be integrated with other advanced composition accounting algorithms applicable to the kernel DP mechanism. This integration results in a final expression that manifests as a linear combination with specific coefficients, as delineated in Theorem 3. The detailed accounting process for this integration can be methodically extracted from Algorithm 3.

 \begin{algorithm}[t]
\caption{Composition accountant for BR-DP }
\hspace*{\algorithmicindent}
\textbf{Input:} $q$, $\epsilon_y$, $\delta_y$, $T$, $\Delta_f$, target $\epsilon$.\\
 \hspace*{\algorithmicindent} \textbf{Output:} $\delta_{\Gamma}^T(\epsilon)$.
\begin{algorithmic}[1]
\State $\{\tau_l, \tau_u, \tau'_l, \tau'_u\}$ $\gets$ $\Delta_f$ and $\theta$;
\State $\Phi_{(\epsilon_y, \delta_y)}(\cdot)$ $\gets$ $\epsilon_y$, $\delta_y$;
\State $p_{\theta}$ $\gets$ \eqref{eq:pytheta}, $\bar{p}_{\theta}$  $\gets$ \eqref{eq:barpytheta};
\State $W$ $\gets$ \eqref{eq:W}, $\mathcal{L}$ $\gets$ \eqref{eq:l};
\State $\delta_Z^T(\epsilon)$ $\gets$ privacy profile ($\epsilon_y$, $\delta_y$, $T$) such as \cite{Balle2018ImprovingTG}; 
\State initialize $\delta_{\Gamma}^T = 0$;
\For{$1\le k \le T$}
\State $\delta' = \delta^T_{Z}(\epsilon-(T-k)\mathcal{L} $;
\State $\delta_{\Gamma}^T \gets \delta_{\Gamma}^T + \binom{T}{k} (1-W)^kW^{T-k}\delta'$;
\EndFor
\State Return $\delta_{\Gamma}^T$
\end{algorithmic}
\end{algorithm}

\begin{rmk}
Algorithm 3 exhibits a linear computational complexity with respect to $T$, denoted as $\mathcal{O}(T)$.
\end{rmk}

We present a numerical comparison for BR-DP and DP composed leakage with Gaussian and Laplacian kernel, with results shown in Fig. \ref{compo_gau} and Fig. \ref{comp_lap}, respectively. For each plot, we let $\epsilon = 1$, $\theta = 1$, $\Delta_f = 1$. Then $\epsilon_y$ and $q$ are derived from Algorithm \ref{algo:2} and Algorithm 1 respectively. For the privacy profile of the Gaussian kernel and Gaussian DP mechanism, we deploy the analytic Fourier Accounting algorithm proposed in \cite{Zhu2021OptimalAO}; For the privacy profile of the Laplacian kernel and Laplacian DP mechanism, we apply RDP composition accounting algorithm in \cite{8049725}. We observe that BR-DP framework effectively reduces composed leakage compared to conventional DP mechanisms.

\begin{figure}[t]
\centering 
\subfigure[Composition leakage with Gaussian kernel, composed kernel DP leakage measured by Analytic Fourier Accounting in \cite{Balle2018ImprovingTG}.]
{\includegraphics[width=0.4\textwidth]{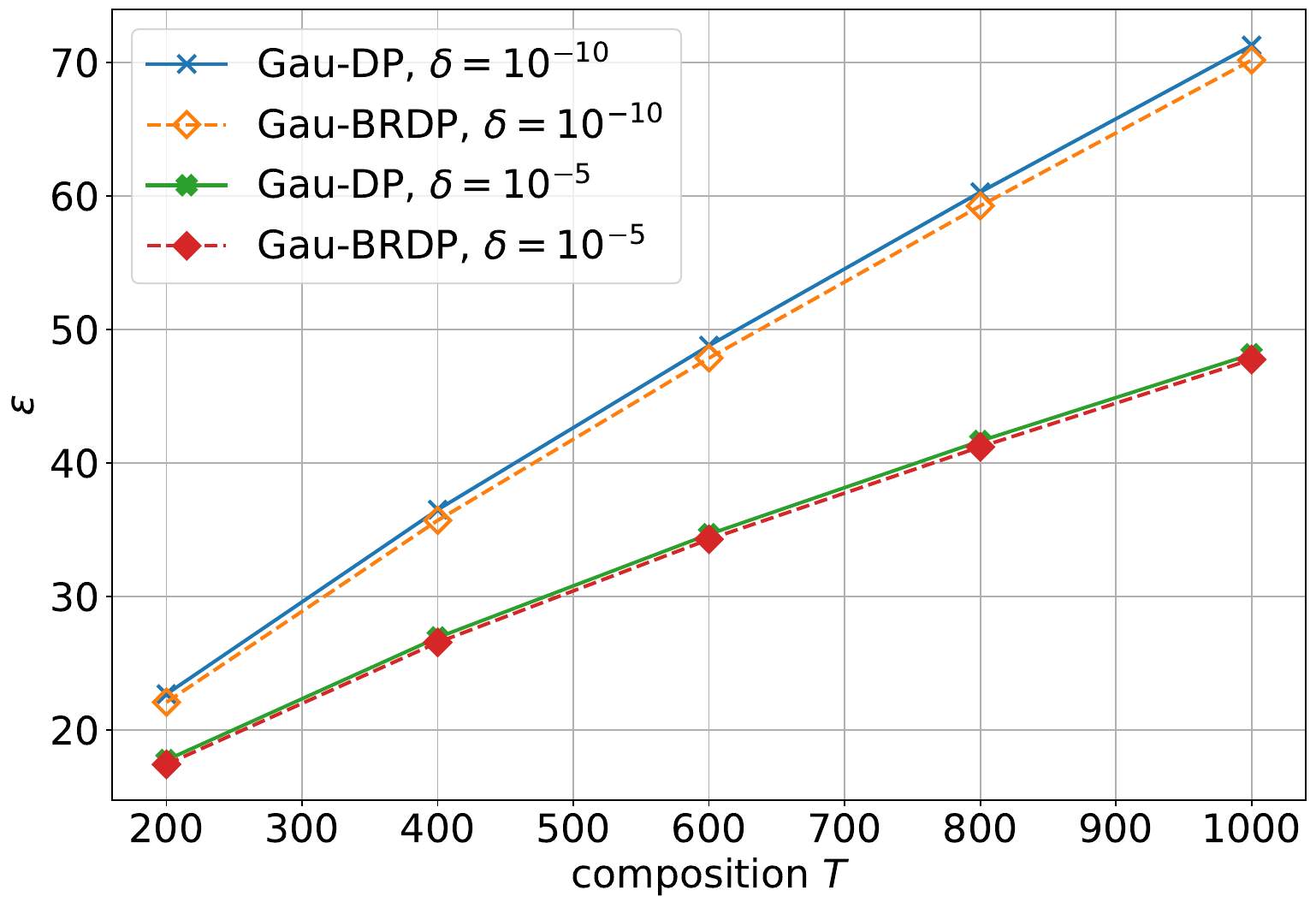}
\label{compo_gau}}
\subfigure[Composition leakage with Laplacian kernel, composed kernel DP leakage measured by RDP in \cite{8049725}.]
{\includegraphics[width=0.41\textwidth]{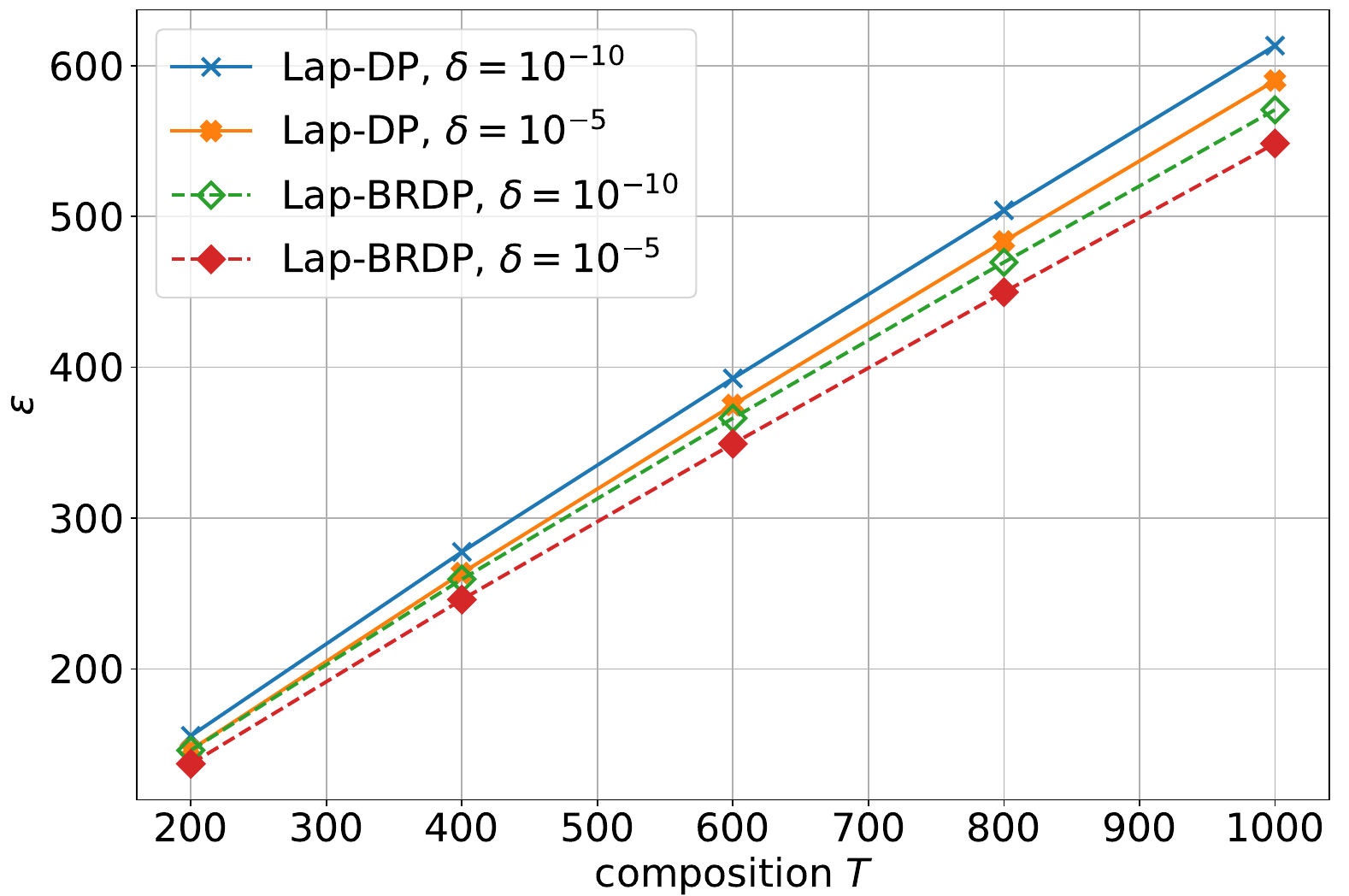}
\label{comp_lap}}
\caption{Composed leakage comparison for BR-DP and DP with Gaussian kernel and Laplacian kernel, respectively. For each plot, $\epsilon = 1$, $\theta = 1$, $\Delta_f = 1$. $\epsilon_y$ and $q$ are derived from Algorithm 3 and Algorithm 2, respectively.}
\end{figure}

\section{Subsampling for BR-DP}



As illustrated in Fig. \ref{ar_vs_ep}, the BR-DP consistently surpasses the DP in \textit{acceptance rate}. Nevertheless, BR-DP may require a large \(\epsilon\) budget to elevate the \textit{acceptance rate} beyond a specific threshold. For instance, as depicted in Fig.\ref{ar_vs_ep}, BR-DP needs \(\epsilon > 4.5\) to achieve an \textit{acceptance rate} over \(0.9\). However, under stringent privacy regulations or for practical use, large \(\epsilon\) values are often not permitted. To mitigate this, we investigate the privacy amplification by subsampling. This technique provides deniability for each individual by randomizing individual inclusion in the dataset for data aggregation, thereby enhancing privacy. In BR-DP, the strategy involves starting with a small \(\epsilon\) budget and amplifying it to a larger effective budget \(\epsilon'\) through aggregating a data subset. This method facilitates utility improvement while adhering to the amplified \(\epsilon'\).

\subsection{Privacy Amplification by Subsampling}

The following Theorem states the privacy amplification by subsampling for BR-DP.
\begin{thm}[Privacy Amplification by Subsampling for BR-DP]
    Let $X_{s}$ be a random subset from $X$ that $|X_s|<|X|$. If a BR-DP framework $\mathcal{M}$ is $(\epsilon,\delta)$-DP for $X$, then the subsampled $\mathcal{M}^s$ takes as input of $X_s$, is $(\epsilon',\delta')$-DP, where $\epsilon' = \log[(e^{\epsilon}-1)p + 1]$, $\delta' = p\delta$, and $p$ denotes the sampling rate: $|X_s|/|X|$.
\end{thm}
The derivation of this proof parallels the methodology used in the proof for DP subsampling in \cite{10.5555/3327345.3327525}. To maintain brevity, we omit the detailed proof here. Also note that different sampling strategies such as Poisson sampling and fixed-size sampling, achieve slightly different amplification factors. However, since the detailed sampling technique is not the main focus of this paper, we adopt a general form of $p = |X_s|/|X|$ as the sampling rate. 

It is noteworthy that while the sampling phase amplifies data privacy, it concurrently bases subsequent processes on a subset of the raw dataset. This selective usage inherently impacts the overall data utility. Typically, a higher sampling probability (\( p \)) enhances the accuracy of aggregation, yet necessitates the introduction of substantial noise during the perturbation phase as the privacy budget is not sufficiently amplified. Therefore, the added noise in perturbation becomes the predominant factor in utility degradation. Conversely, a lower \( p \) allows for the injection of lesser noise during data perturbation. In this scenario, the loss of information during the sampling phase becomes the primary source of error. Intuitively, this establishes a trade-off between information loss in the sampling phase and the perturbation phase. There exists an optimal sampling probability \( p \) that achieves the most favorable balance, effectively minimizing the overall error.

To derive the optimal sampling rate, we note that error caused by a BR-DP framework with subsampling is dependent on the data distribution and the query type, in the following, we assume that each individual's data is sampled from a Gaussian distribution. We also consider the three most common types of query: counting query, average query, and summation query. For each query type, we first represent the error in the sampling phase as a random variable $\mathcal{E}$, and then provide the accuracy analysis with BR-DP framework parameters. For ease of conveying our ideas, we assume the BR-DP framework is based on a Gaussian kernel DP mechanism.

\subsection{Analyzing Utility Across Various Query Types}

In the following, we denote $x_i$ as an individual's data in the dataset, where $i$ denotes the individual's index in the dataset. As mentioned above, we assume each $x_i$ is sampled from a Gaussian distribution with mean $\mu$, and a standard deviation of $\sigma$. Denote $\mathcal{T}$ as the set of users' index in the original dataset. i.e., $X= \{x_i\}_{i\in\mathcal{T}}$, and denote $\mathcal{S}$ as the subset of users' index in the randomly sampled dataset. i.e., $X_S= \{x_i\}_{i\in\mathcal{S}}$. We next express the error caused by the subsampling for different queries as noisy random variables, we provide detailed query function expressions and subsampled results in the appendix.

The noise introduced in the subsampling phase can be derived as:
\begin{equation*}
\begin{aligned}
    &\mathcal{E}_{\text{sum}} = \sum_{i\in{\mathcal{T}}}x_i - \frac{1}{p}\sum_{i\in{\mathcal{S}}}x_i,\\
    &\mathcal{E}_{\text{avg}} = \frac{1}{|X|}\sum_{i\in{\mathcal{T}}}x_i - \frac{1}{|X_s|}\sum_{i\in{\mathcal{S}}}x_i,\\
    &\mathcal{E}_{\text{cnt}} = \sum_{i\in{\mathcal{T}}} \mathbbm{1}_{\{x_i\in \mathcal{C}\}} - \frac{|X|}{|X_s|}\sum_{i\in{\mathcal{S}}} \mathbbm{1}_{\{x_i\in \mathcal{C}\} },
\end{aligned}
\end{equation*}
for summation query, average query, and for counting query, respectively. Note that, for counting queries, $\mathcal{C}$ denotes the subset of each individual's input support, such as $x_i$ is counted as $1$, iff $x_i \in \mathcal{C}$.

The following proposition describes the distribution of $\mathcal{E}$ for different types of queries.

\begin{prop}
    The noisy random variable $\mathcal{E}$ for the three types of queries above are zero-mean Gaussian distribution, with standard deviations of $\sigma_{\text{sum}}= \sigma_x \sqrt{|X|(1-p)/p}$ for summation query, $\sigma_{\text{avg}}= \sigma_x \sqrt{(1-p)/(p|X|)}$ for averaging query, and $\sigma_{\text{cnt}}=\sqrt{|X|(1-p)p_c(1-p_c)/p}$ for counting query, where $p_c = \operatorname{Pr}(x_i\in{\mathcal{C}})$.
\end{prop}

Observe that for a given sampling rate $p$, $\sigma_{\text{sum}}$ and $\sigma_{\text{cnt}}$ increases with $|X|$ while $\sigma_{\text{avg}}$ decreases with $|X|$. On the other hand, for a given $|X|$, the standard deviations of all three types of query are proportional to $\sqrt{(1-p)/p}$, which monotonically decreases with $p$.

\subsection{Optimal Subsampling Rate}

Combined with subsampling, the privacy-protection mechanism can be updated by appending the noise from subsampling, and the noise from the perturbation to the raw answer $Y$:
\begin{equation}
    Y_n = Y + \mathcal{E} + N,
\end{equation}
where $\mathcal{E}$ denotes the noise calibrated in the subsampling phase and $N$ represents the noise injected by the DP kernel with an amplified privacy budget. We next tackle the dilemma of the optimal sampling rate $p$ for a BR-DP framework with a Gaussian kernel DP.

 \begin{algorithm}[t]
\caption{Find Optimal $p$ for Subsampled BR-DP with Gaussian Kernel}
\hspace*{\algorithmicindent}
\textbf{Input:} $\epsilon$, $\delta$, $\Delta_f$, $N$, $\sigma_x$, tol.\\
 \hspace*{\algorithmicindent} \textbf{Output:}$p$.
\begin{algorithmic}[1]
    \State $p_{low} \gets 0$;
    \State $p_{up} \gets 1$;
    \While{$p_{up} - p_{low} > \text{tol}$}
        \State $p_1 \gets p_{low} + \frac{p_{up} - p_{low}}{3}$;
        \State $p_2 \gets p_{up} - \frac{p_{up} - p_{low}}{3}$;
        \State $\epsilon_1, \delta_1$ ($\epsilon_2, \delta_2$) $\gets$ theorem 4 with $p_1$ ($p_2$);
        \State calibrate  $\epsilon_1, \delta_1$ ($\epsilon_2, \delta_2$) to  $\epsilon'_1, \delta$ ($\epsilon'_2, \delta$);
        \State $\sigma_{\mathcal{E}1}$  ($\sigma_{\mathcal{E}2}$) $\gets$ prop. 7 with $p_1$ ($p_2$);
        \State $\epsilon_{y1}$ ($\epsilon_{y2}$) $\gets$ Algo. 2 with $\epsilon_1,\delta_1$ ($\epsilon_2,\delta_2$), $O$ $\gets$ \eqref{eq:budget2};
        \State $q_1$ ($q_2$) $\gets$ Algo. 1 with $\epsilon_1, \delta_1$ ($\epsilon_2, \delta_2$);
        \If{$O(\epsilon_{y1}, q_1, p_1)> O(\epsilon_{y2}, q_2, p_2)$)}
            \State $p_{low} \gets p_1$;
        \Else
            \State $p_{up} \gets p_2$;
        \EndIf
    \EndWhile
    \State $p \gets (p_{up} + p_{low})/ 2$;
    \State \Return $p$
\end{algorithmic}
\end{algorithm}

\begin{figure}[t]
\centering 
\subfigure[Optimal sampling rate as a factor of the cardinally of the input dataset.]
{\includegraphics[width=0.4\textwidth]{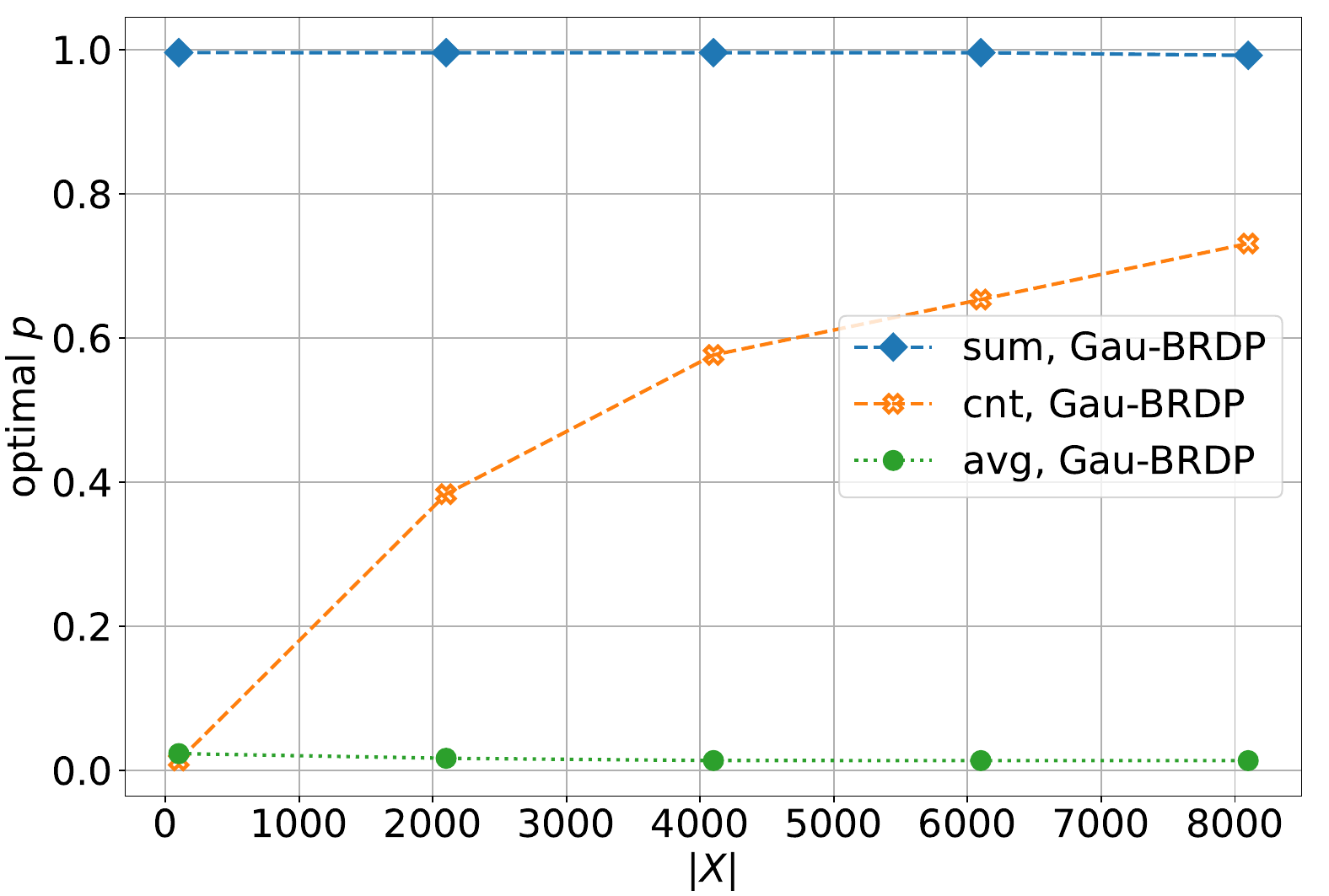}
\label{p_vs_N}}
\subfigure[Acceptance rate as a factor of the cardinally of the input dataset.]
{\includegraphics[width=0.4\textwidth]{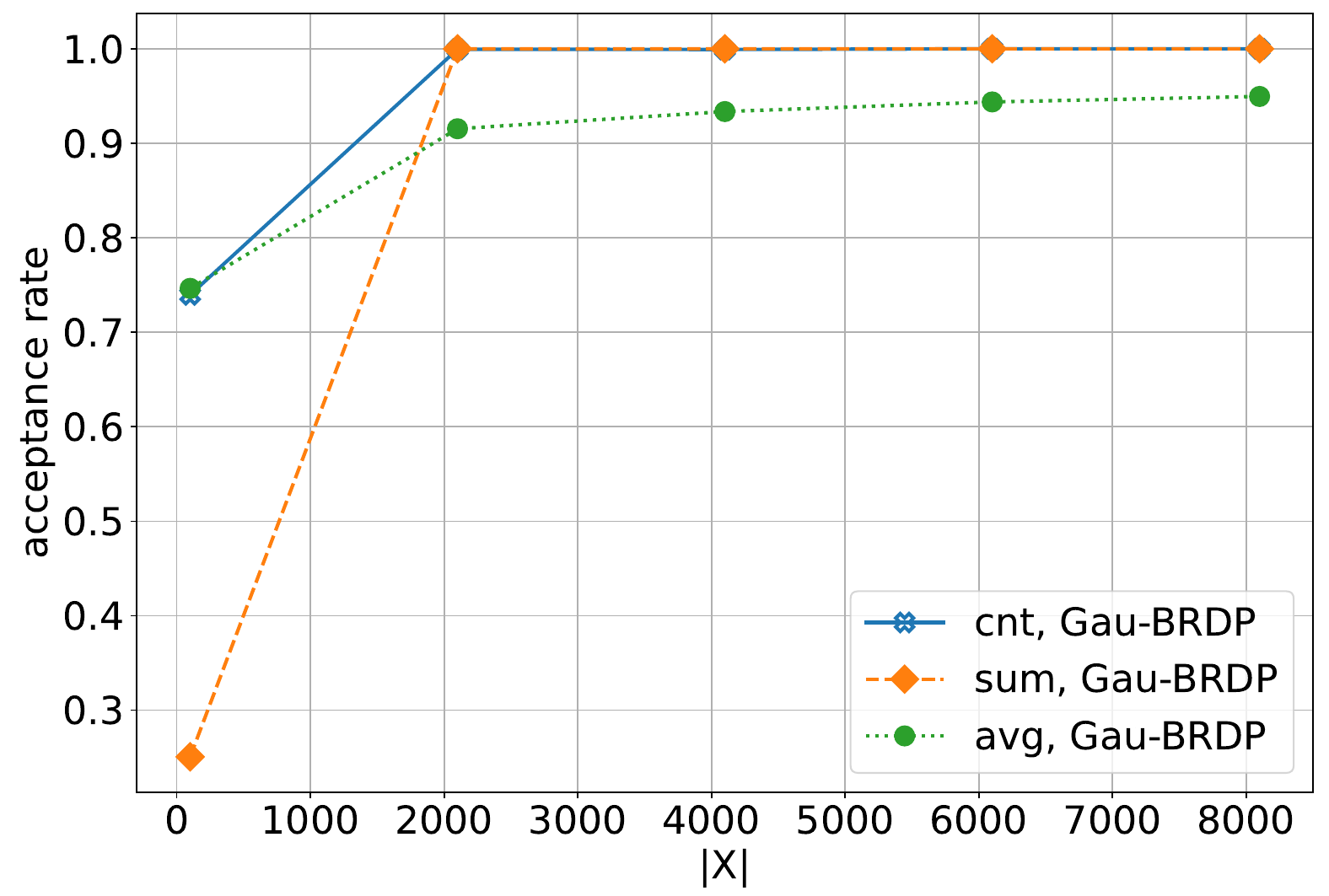}
\label{ar_vs_N}}
\caption{Subsampling evaluation for different mechanisms, $\sigma_x = 10$, $\epsilon = 0.1$, $\delta = 10^{-5}$, $p_c = 0.1$}
\vspace{-10pt}
\end{figure}

The next proposition presents the updated utility function with subsampling.

\begin{figure*}[htp]
\centering 
\subfigure[Summation query for Adult ($\Delta_f = 22, \theta = 5$ )]
{\includegraphics[width=0.45\textwidth]{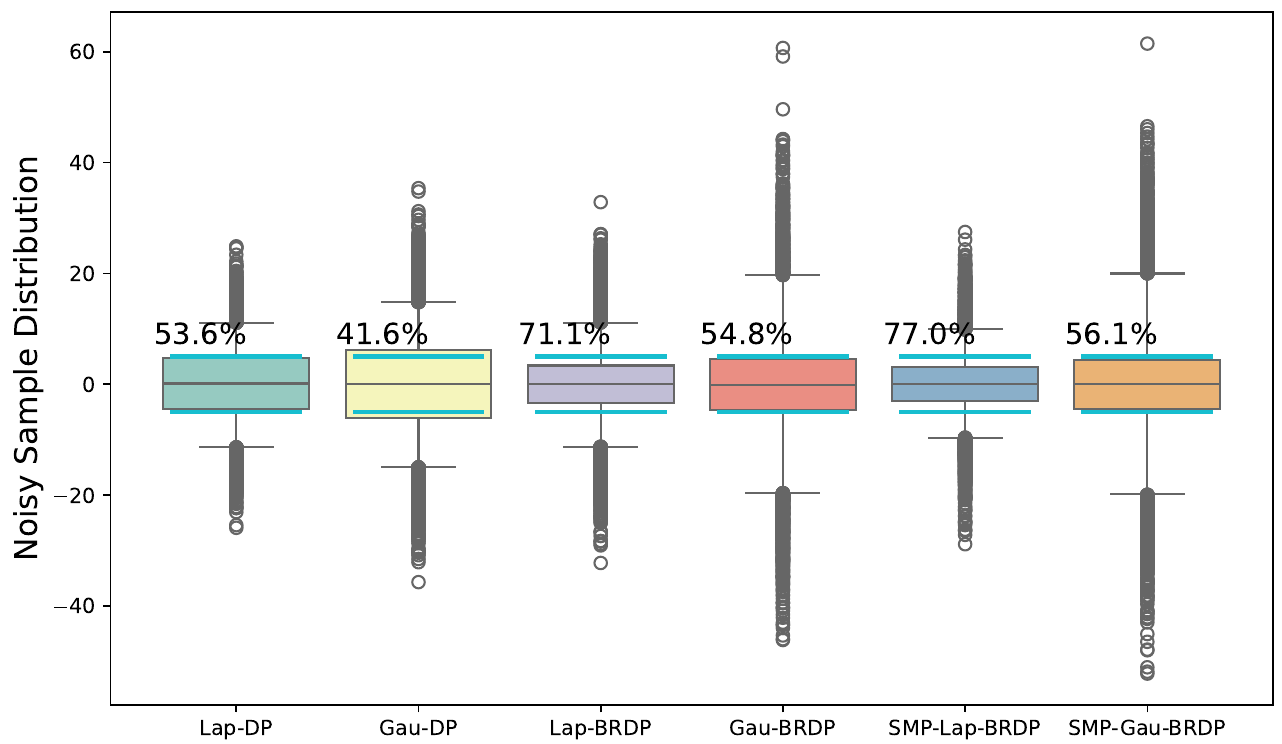}
\label{fig:sum1}}
\subfigure[Summation query for Adult ($\Delta_f = 12, \theta = 10$)]
{\includegraphics[width=0.39\textwidth]{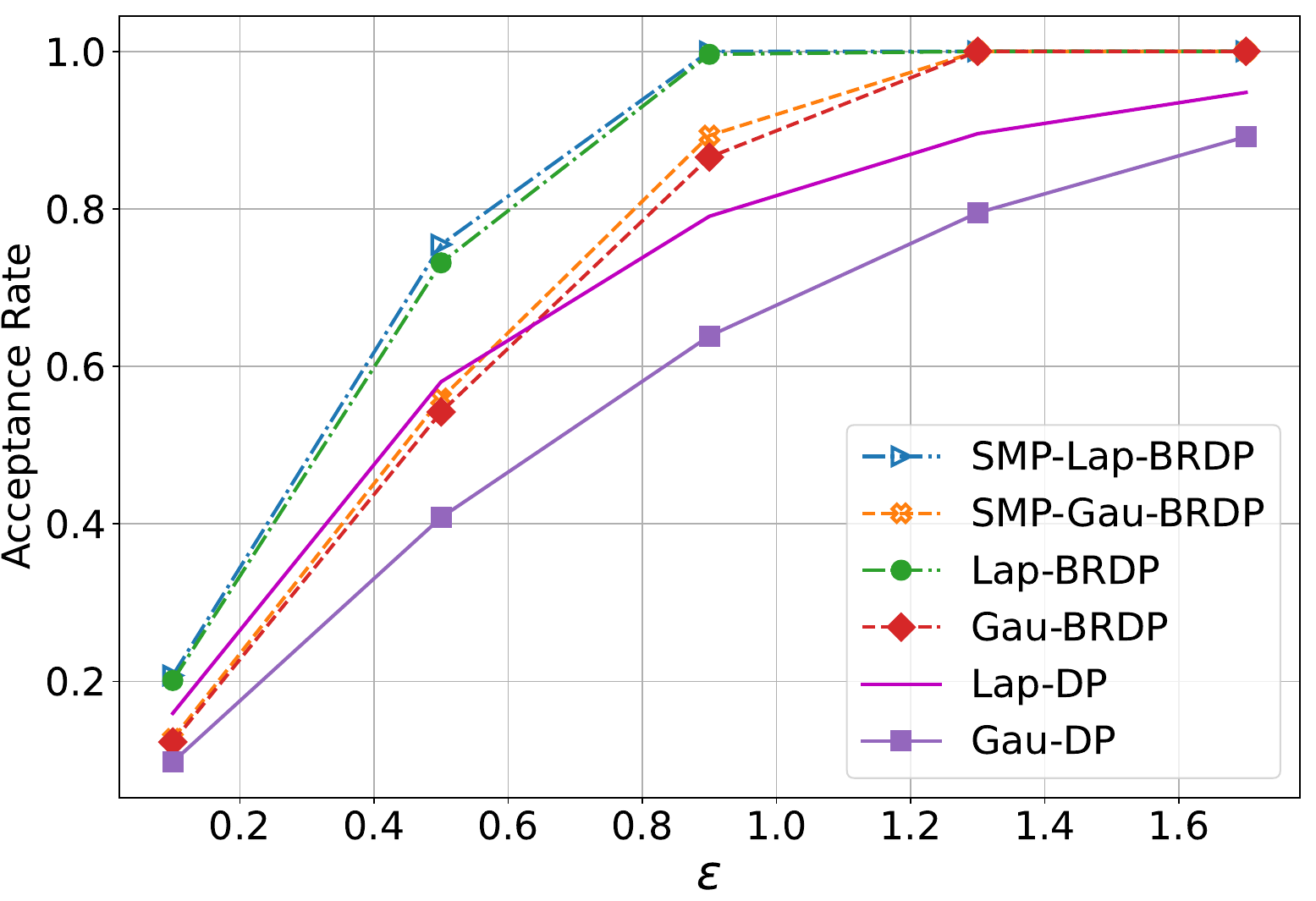}
\label{fig:sum2}}
\subfigure[Averaging query for Adult ($\Delta_f = 0.0017,\theta =1$)]
{\includegraphics[width=0.45\textwidth]{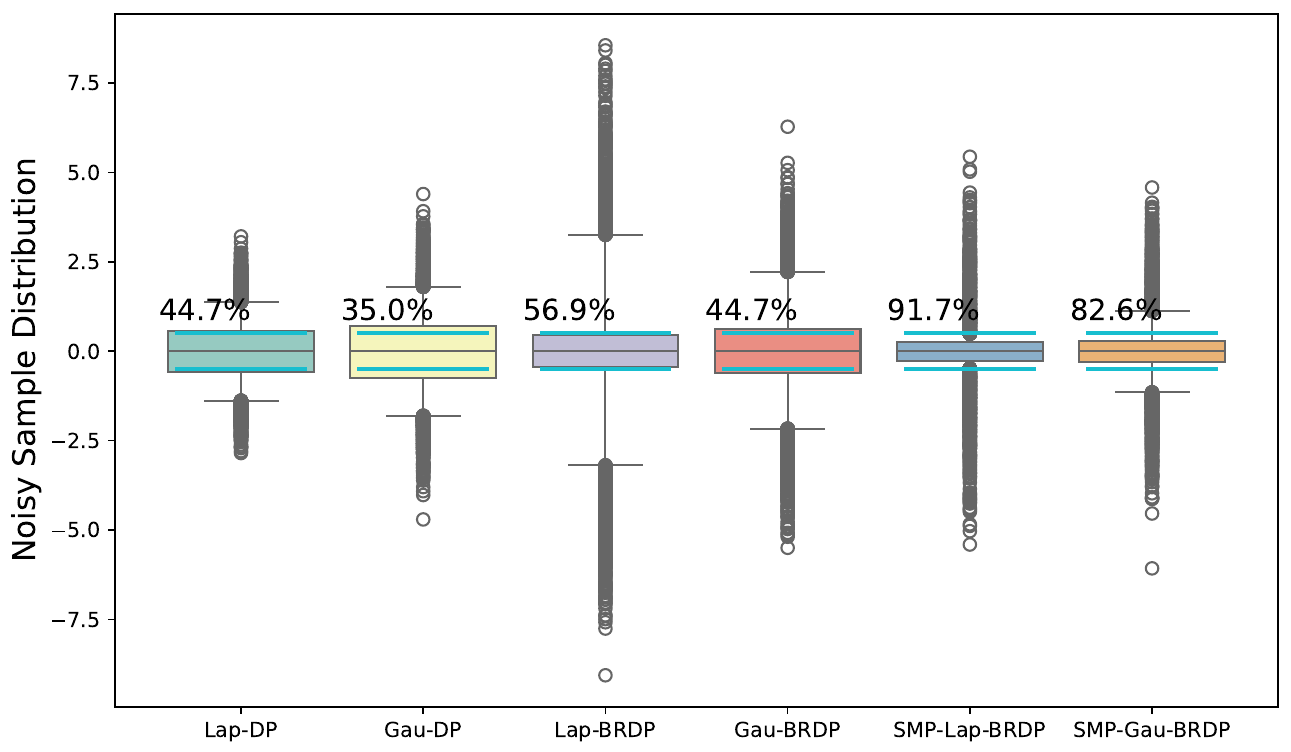}
\label{fig:avg1}}
\subfigure[Averaging query for Adult ($\Delta_f = 0.0012,\theta =1$)]
{\includegraphics[width=0.39\textwidth]{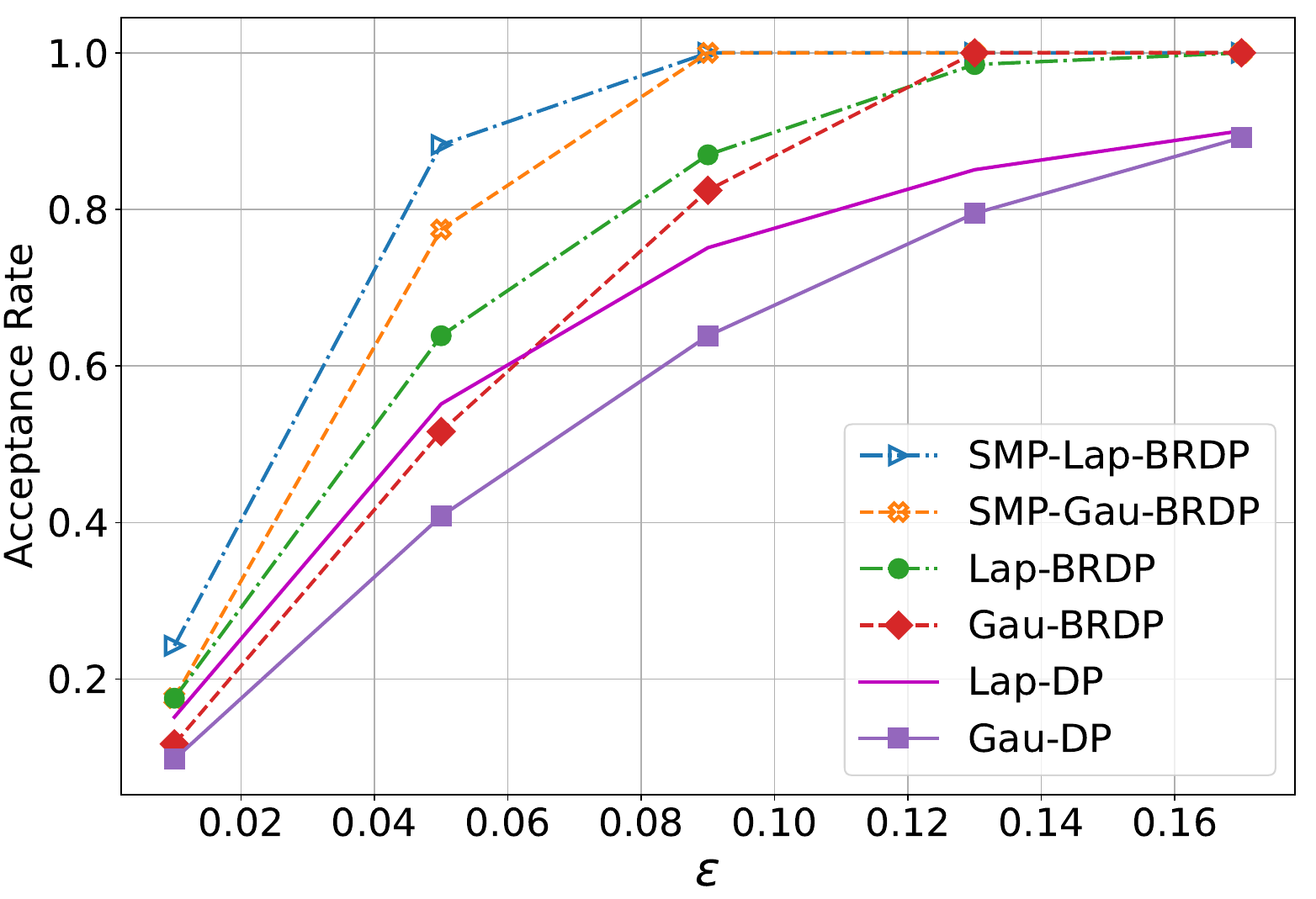}
\label{fig:avg2}}
\subfigure[Counting query for Adult ($\Delta_f = 1, \theta=3$)]
{\includegraphics[width=0.45\textwidth]{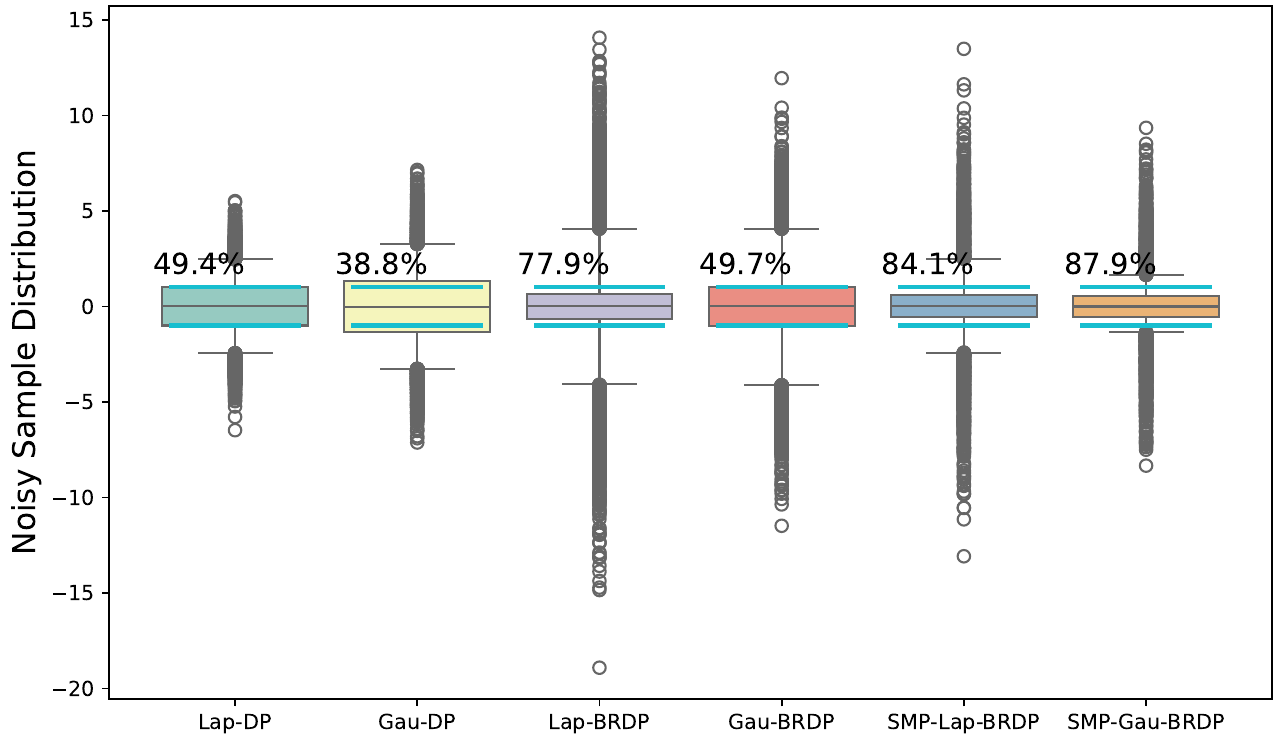}
\label{fig:cnt1}}
\subfigure[Counting query for Adult ($\Delta_f = 1, \theta = 3$)]
{\includegraphics[width=0.39\textwidth]{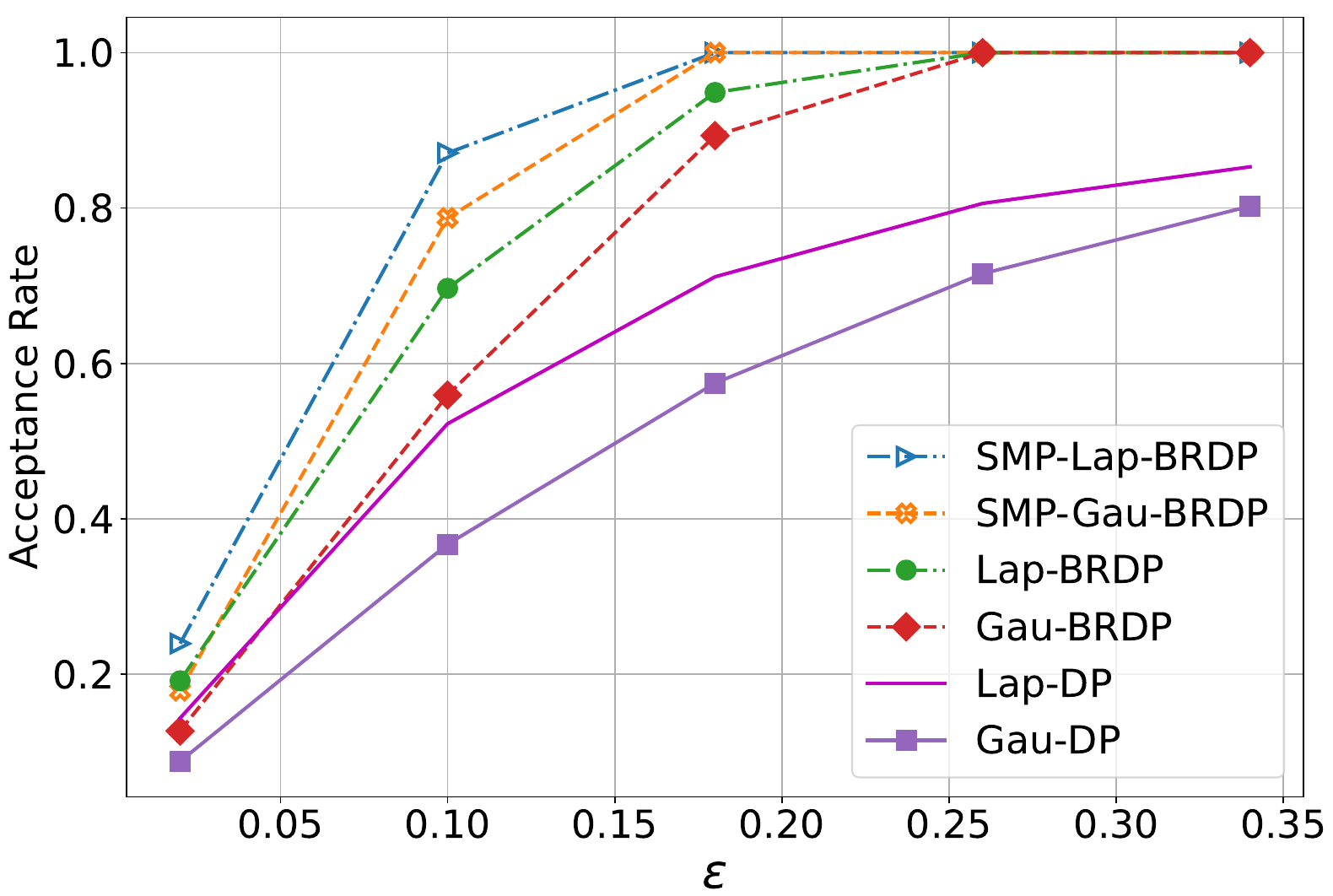}
\label{fig:cnt2}}
\caption{Data generation of Gaussian/Laplacian kernel BR-DP frameworks with DP with Adult datasets. For summation query $(\epsilon=0.5, \delta=10^{-5})$, for averaging query $(\epsilon=0.05, \delta=10^{-5})$, and for counting query $(\epsilon=0.1, \delta=10^{-5})$.}
\label{fig:exp0}
\end{figure*}

\begin{prop}
    For a subsampled BR-DP mechanism, with $\mathcal{E}\sim \mathcal{N}(0, \sigma_{\mathcal{E}})$ and a BR-DP mechanism with $(0,\sigma_y)$ Gaussian kernel. The acceptance rate defined in \eqref{eq:budget} is updated to
    \begin{equation}\label{eq:budget2}
O(\epsilon_y,q,p) = 1-q+\frac{q}{\Phi_{\sigma}(\tau_{u})-\Phi_{\sigma}(\tau_{l})},
\end{equation}
where $\Phi_{\sigma}$ stands for the CDF of a $\mathcal{N}(0,\sigma^2)$ Gaussian distribution, with $\sigma = \sqrt{(\sigma_y^2 + \sigma_{\mathcal{E}}^2)}$.
\end{prop}

\begin{proof}
Denote $N_{tot}$ as the total noise injected by subsampling and BR-DP. Then
\begin{equation*}
    f_{N_{tot}} = f_{\mathcal{E}}\ast f_{brdp} = f_{\mathcal{E}}\ast f_{N} \ast f_{R}.
\end{equation*}
As both $\mathcal{E}$ and $N$ are Gaussian, the convolution of $f_{\mathcal{E}}\ast f_{N}$ is also Gaussian with a standard deviation of $\sigma = \sqrt{(\sigma_y^2 + \sigma_{\mathcal{E}}^2)}$. The other steps for deriving the expression in \eqref{eq:budget2} follow similar ideas as in the steps for \eqref{eq:budget}. 
\end{proof}
The following Algorithm 4 provides steps to obtain optimal parameters of $\epsilon_y$, $q$ and $p$.

We present numerical findings on the optimal sampling rate and the corresponding acceptance rate for the Gaussian kernel BR-DP framework with optimal subsampling. The results are depicted in Figures \ref{p_vs_N} and \ref{ar_vs_N}. For the summation query, the optimal sampling rate approaches $1$, as the function $f_{\mathcal{E}}$ predominantly contributes to inaccuracies in the final result for large $\sigma_x$, necessitating a higher $p$ to reduce subsampling noise. Conversely, the average query favors a lower sampling rate because the variance of $\sigma_{\text{avg}}$ is small for large datasets, compared to $\sigma_N$ in perturbation. The optimal sampling rate for the counting query varies with the dataset size; for smaller datasets, $\sigma_N$ primarily drives the error, whereas for larger datasets, the increasing $\sigma_{\mathcal{E}}$ becomes the dominant error factor. The key observation from Fig. \ref{ar_vs_N} is that the \textit{acceptance rate} also benefits from a large input dataset. Notably, the above analysis is derived based on parameters set as shown in the figure. Different parameters may lead to different conclusions. Also, note that similar analysis also applies to subsampled DP mechanisms and can be useful in deriving the optimal sampling rate from the utility perspective.

\section{Experiments}

In this section, we experiment on real datasets to show the advantage of the BR-DP mechanism. We first show that BR-DP is efficient in improving the acceptance rate. The evaluations focus on three query types: summation query, average query, and counting query. Then we compare the privacy leakage of BR-DP and DP after composition.

\begin{table*}[t]
    \centering
\caption{Privacy Leakage {and runtime} After $1000$ times Composition for $\epsilon = 0.1$.}
\label{tab:my_label}
\renewcommand{\arraystretch}{1.5}
    \begin{tabular}{c|c|c|c|c|c|c} 
          &Lap-DP&  Gau-DP&  Lap-BRDP&  Gau-BRDP&  SMP-Lap-BRDP&  SMP-Gau-BRDP\\ \hline 
          $\delta = 10 ^{-5}$ &18.98& 4.77 &  16.8&  4.72 &  49.37 & 13.13\\ \hline 
          $\delta = 10 ^{-10}$ &25.38&  7.17&  22.23& 6.93& 65.6&  19.6\\ \hline 
          {Runtime}& {12.3}& {13.5} &  {63.2}&  {77.5} &  {4034.5} & {5162.2}
    \end{tabular}
\end{table*}


\subsection{Experiment Settings}

\textbf{Datasets: } {we} use two datasets that are commonly used in the privacy research community with different focuses. 

\textit{Adult Dataset}\cite{misc_adult_2} from the UCI, which contains census information with $45,222$ records and $15$ attributes. The attributes include both categorical ones such as race, gender, and education level, as well as numerical ones such as capital gain, capital loss, and weight. 

\textit{Gowalla} \cite{cho2011friendship} a location-based social networking website where users share their locations by checking in. Gowalla includes a total of $6,442,890$ check-ins of these users over the period of Feb. $2009$ - Oct. $2010$. The attributes include user-id, check-in time, location-longitude, location-latitude, and location-id.

\textbf{Methodologies} We compare BR-DP mechanisms mentioned in this paper, with DP, and We consider both Gaussian and Laplacian noise in the kernel for BR-DP and for DP respectively. 
We evaluate different mechanisms with three types of queries, i.e., summation, average, and counting. We first preprocess each dataset, remove non-available data, and randomly select a subset of $10,000$ individuals for each dataset. Then we divide each of the two datasets into $10$ disjoint subsets according to the user ID. Then we submit identical queries to each subset for $1000$ times and showcase the statistics of the noisy outputs on a uniformed plane by calibrating the average of output to $0$. We are interested in the aggregation of growth income of each subset for Adult dataset; and the total check-in duration of each identical individual in each subset for Gowalla. Specifically, we compare the acceptance rate provided by different mechanisms. Note that the calibration does not violate the privacy guarantee due to the post-processing property. Also averaging over multiple disjoint subsets does not degrade privacy protection, as subsets are divided to be disjoint with each other, and parallel composition applies.

\subsection{Noisy Answers and Utility-Privacy Tradeoff}
We first showcase the data generation scenarios for different mechanisms under different queries. The results are shown in Fig. \ref{fig:sum1}, \ref{fig:avg1}, and \ref{fig:cnt1}, respectively. Corresponding parameters are listed along with the figures. In the box-plot. The area of the box covers $50\%$ of the generated data, and the region bounded by the upper and lower dash represents $95\%$ of data distribution. We also marked the utility boundaries on the plot together with the calculated ``acceptance rate''. When selecting the subsample rate, we first simulate with estimated $\sigma_x$ and the size of each subset. We adopt the same sampling rate derived from the subsampled Gaussian mechanism for the subsampled Laplacian mechanism. For the scale parameter $b$ in Laplacian mechanism, we first calibrate $(\epsilon,\delta)$-Laplacian to $\epsilon'$-Laplacian with enhanced $\epsilon'$. Then $b = \Delta_f/\epsilon'$. We then show the acceptance rate of different mechanisms discussed above as a parameter of the total budget $\epsilon$. Again, for each $\epsilon$, we run the identical queries for $1,000$ times to get the average acceptance rate for each mechanism. The corresponding results are shown in the column to the right of Fig. \ref{fig:exp0}.

Observe that with both Laplacian and Gaussian kernels, BR-DP always outperforms the corresponding DP mechanisms given the same privacy budget. Moreover, BR-DP equipped with optimal subsampling even enhances the utility. In general, Laplacian kernel BR-DP outperforms the Gaussian kernel BR-DP mechanism, this is due to the geometry of the noise distribution: the Laplacian noise is more concentrated than the Gaussian noise under the same privacy budget and sensitivity. 

\subsection{{Privacy Leakage and Runtime Comparison with Composition}}
In the following, we compare the privacy leakage {and runtime} of different mechanisms described above after composition. Since the experimental results are averaged over $1,000$ independent and identical queries, the privacy leakage for each individual query is enlarged, and the increased amount can be measured by the composition theorem. 

From Section \ref{sec:tightcompo}, we know that the $T$-fold composition of independent BR-DP mechanisms can be viewed as a linear transformation of the corresponding kernel DP mechanisms after $T$-fold composition. Here, we adopt the technique of the ``analytic Fourier Accounting" algorithm in \cite{Balle2018ImprovingTG} to account for the composition leakage of the Gaussian mechanism, which is proven to be the state-of-the-art tightest analysis. For the Laplacian mechanism, we consider a Renyi Differential Privacy approach (RDP) \cite{8049725}. For the subsampled BR-DP mechanism, we use the sampling rate to get an amplified privacy budget. Then we calculate the composed leakage with the amplified budget and then calibrate the final leakage with the sampling rate. The composed leakages for different mechanisms are shown in Table 1. Observe that the Gaussian BR-DP incurs the smallest leakage compared to other mechanisms. It is worth noting that while the Laplacian mechanism provides better utility compared to the Gaussian mechanism as a BR-DP kernel, it suffers from additional privacy leakage after composition. {On the other hand, we calculate the runtime by each mechanism. Since $\delta$ only makes slight impact on the runtime, we intentionally ignore it for simplicity. }

\section{Limitations and Future Works}
We identified two limitations of BR-DP. The first limitation of BR-DP is that its utility may only be equivalent to that of traditional DP, depending on the allocated privacy budget and query sensitivity. BR-DP generally shows greater effectiveness for queries with small answer values but high sensitivity. Nevertheless, BR-DP can be considered an adaptive framework, as it optimally allocates the privacy budget. In cases where the optimal outcome aligns with standard DP performance, BR-DP seamlessly reverts to its underlying DP kernel mechanism.
The second limitation of BR-DP stems from its tendency to increase the variance of noisy outputs, thereby exacerbating the inaccuracy of out-of-bound responses. This trade-off is the sole sacrifice required from BR-DP users to ensure stringent privacy protection while enhancing utility. This approach is particularly viable when out-of-bound answers are entirely unacceptable. Notably, in practical applications, it's possible to set \(\theta\) as unbounded, which effectively degrades the BR-DP framework to its underlying DP kernel mechanism.

In terms of future work, a direct future work of BR-DP is to relate the definition with RDP, from there, more tightened results about privacy amplification with composition can be derived. Another direction is to adapt the BR-DP framework to local settings, where each individual perturbs and releases data without a server, the BR-DP framework is useful in providing closed-form perturbation parameters for discrete-valued mechanisms such as generalized randomized response and may contribute to optimizing the frequency estimation oracles proposed in\cite{203872} or context-aware mechanisms in \cite{LIP2}. It might also be of interest to investigate more applications of BR-DP rather than the querying system. One possible direction is to combine BR-DP with deep learning: on one hand, BR-DP has the potential to generate a more stable gradient than DP, on the other hand, BR-DP incurs less privacy leakage post-composition.

\section{Conclusion}
In this paper, we introduce and examine Budget Recycling Differential Privacy (BR-DP), a novel framework designed to enhance the utility of traditional differential privacy (DP) mechanisms. The utility is quantified through the \textit{acceptance rate}, defined as the probability of noisy outputs falling within a predetermined error boundary. Central to the BR-DP framework is a kernel DP mechanism responsible for generating noisy outputs, accompanied by a recycler that reprocesses the budget to regenerate outputs if they are deemed \textit{unacceptable}. We conduct a thorough analysis of privacy leakage and present algorithms for the optimal determination of parameters. Our study extends to a rigorous composition analysis of BR-DP, complete with an accounting algorithm. Additionally, we explore the concept of privacy amplification through subsampling and propose an algorithm to determine the optimal sampling rate, further enhancing data utility. Evaluations conducted on real datasets demonstrate BR-DP's superiority over conventional DP mechanisms. Our results indicate a substantial increase in the \textit{acceptance rate} and a reduced privacy leakage post-composition.

\bibliographystyle{IEEEtran}
\bibliography{ref}

\appendices

\section{Proof of BR-DP noisy distribution}
\begin{proof}
The mechanism releases a $y_n$ satisfies $||n||\le\theta$ with  probability of:
\begin{equation*}
    \begin{aligned}
&f_N(y_n|y)+f_N(y_n|y)\text{Pr}(||N||>\theta)q\\
&~~~~~~~~~~~+f_N(y_n|y)\text{Pr}(||N||>\theta)^2q^2+...\\
=&f_N(y_n|y)+f_N(y_n|y)\bar{p}_{\theta}q+f_N(y_n|y)\bar{p}_{\theta}^2q^2+...\\
=&f_N(y_n|y)\cdot\sum_{k=0}^{\infty}\bar{p}_{\theta}^kq^k
=\frac{f_N(y_n|y)}{1-\bar{p}_{\theta}q},
\end{aligned}
\end{equation*}

The probability representation described above can be interpreted as follows: an acceptable $y_n$ is generated in the first round with a probability density of $f_N(y_n|y)$. The first round is unsuccessful and is recycled with a probability of $\bar{p}_{\theta}q$. In the subsequent second round, a $y_n$ is released again with a probability density of $f_N(y_n|y)$ of meeting the criteria. This process continues indefinitely, potentially up to an infinite number of rounds.

On the other hand, the mechanism releases an unacceptable $y_n$ with a probability of:
\begin{equation*}
    \begin{aligned}
&f_N(y_n|y)(1-q)+f_N(y_n|y)\text{Pr}(||N||>\theta)(1-q)\\
&~~+f_N(y_n|y)\text{Pr}(||N||>\theta)^2(1-q)^2+...\\
=&f_N(y_n|y)(1-q)+f_N(y_n|y)(1-q)\bar{p}_{\theta}\\
&~~+f_N(y_n|y)\bar{p}_{\theta}^2(1-q)^2+...\\
=&f_N(y_n|y)(1-q)\sum_{k=0}^{\infty}\bar{p}_{\theta}^kq^k
=\frac{f_N(y_n|y)(1-q)}{1-\bar{p}_{\theta}q},
\end{aligned}
\end{equation*}
\end{proof}

\section{Validation of the noise distriution}
\begin{proof}
Obviously,  $0\le(1-q)\le1$, $0\le1-\bar{p}_{\theta}q\le 1$. On the other hand:
\begin{small}
\begin{equation*}
    \begin{aligned}
    &\int_{-\infty}^{\infty}f_{brdp}(y_n|y)dy_n\\
    =&\int_{||n||>\theta}\frac{f_N(y_n|y)(1-q)}{1-\bar{p}_{\theta}q}dy_n+\int_{||n||\le\theta}\frac{f_N(y_n|y)}{1-\bar{p}_{\theta}q}dy_n\\
=&\frac{(1-q)}{1-\bar{p}_{\theta}q}\int_{||n||>\theta}f_N(y_n|y)dy_n+\frac{1}{1-\bar{p}_{\theta}q}\int_{||n||\le\theta}f_N(y_n|y)dy_n\\
=&\frac{(1-q)}{1-\bar{p}_{\theta}q}\bar{p}_{\theta}+\frac{1}{1-\bar{p}_{\theta}q}(1-\bar{p}_{\theta})\\
=&\frac{\bar{p}_{\theta}-q\bar{p}_{\theta}+1-\bar{p}_{\theta}}{1-\bar{p}_{\theta}q}
=1.
\end{aligned}
\end{equation*}
\end{small}
    
\end{proof}

\section{PLD of a BR-DP mechanism}

\begin{proof}
The privacy loss of a BR-DP mechanism can be expressed as:
\begin{equation*}
    \Gamma = \log\left\{\frac{\text{Pr}(\mathcal{M}(Q(X))=y_{n})}{\text{Pr}(\mathcal{M}(Q(X'))=y_{n})}\right\}, 
\end{equation*}
Without loss of generality, we denote $\Gamma$ as the privacy loss caused by the worst-case combination of $X$ and $X'$  as neighboring datasets. We next bound the leakage of the following three cases, with $N'$ denoting the noise generated for $Y'$ generated by a neighboring dataset: 

\textbf{Case 1: }$||N||<\theta$, $||N'||<\theta$:
\begin{equation*}
\begin{aligned}
\frac{\text{Pr}(\mathcal{M}(Q(X))=y_{n})}{\text{Pr}(\mathcal{M}(Q(X'))=y_{n})}
=&\frac{f_{brdp}(y_{n}|y)}{f_{brdp}(y_{n}|y')}\\
=&\frac{f_{N}(y_{n}|y)[1+\bar{p}_{\theta}q+\bar{p}_{\theta}^2q^2+...]}{f_N(y_{n}|y')[1+\bar{p}_{\theta}q+\bar{p}_{\theta}^2q^2+...]}\\
=&\frac{f_N(y_n|y)}{f_N(y_n|y')}.\\
\end{aligned}
\end{equation*}

Then
\begin{small}
\begin{equation*}
    \log\left\{\frac{\text{Pr}(\mathcal{M}(Q(X))=y_{n})}{\text{Pr}(\mathcal{M}(Q(X'))=y_{n})}\right\}=\log\left\{\frac{f_N(y_{n}|y)}{f_N(y_{n}|y')}\right\}.
\end{equation*}
\end{small}

\textbf{Case 2:} $||N||\le\theta$, $||N'||>\theta$:
\begin{equation*}
    \begin{aligned}
&\frac{\text{Pr}(\mathcal{M}(Q(X))=y_{n})}{\text{Pr}(\mathcal{M}(Q(X'))=y_{n})}\\=&\frac{f_{brdp}(y_{n}|y)}{f_{brdp}(y_{n}|y')}\\
=&\frac{f_N(y_{n}|y)[1+\bar{p}_{\theta}q+\bar{p}_{\theta}^2q^2+...]}{f_N(y_{n}|y')(1-q)[1+\bar{p}_{\theta}q+\bar{p}_{\theta}^2q^2+...]}\\
=&\frac{f_N(y_n|y)}{f_N(y_n|y')}\cdot\frac{1}{1-q}.
\end{aligned}
\end{equation*}

\textbf{Case 3:} $||N||>\theta$, $||N'||>\theta$:
\begin{equation*}
\begin{aligned}
&\frac{\text{Pr}(\mathcal{M}(Q(X))=y_{n})}{\text{Pr}(\mathcal{M}(Q(X'))=y_{n})}\\=&\frac{f_{brdp}(y_{n}|y)}{f_{brdp}(y_{n}|y')}\\
=&\frac{f_N(y_{n}|y)(1-q)[1+\bar{p}_{\theta}q+\bar{p}_{\theta}^2q^2+...]}{f_N(y_{n}|y')(1-q)[1+\bar{p}_{\theta}q+\bar{p}_{\theta}^2q^2+...]}\\
=&\frac{f_N(y_n|y)}{f_N(y_n|y')}.
\end{aligned}
\end{equation*}
As Case 3 incurs the same leakage as Case 1, these two cases can be further combined.

Observe that for case 2: 
\begin{equation*}
\begin{aligned}
    &\log\left\{\frac{\text{Pr}(\mathcal{M}(Q(X))=y_{n})}{\text{Pr}(\mathcal{M}(Q(X'))=y_{n})}\right\}
    =&Z + \log\left\{\frac{1}{1-q}\right\},
\end{aligned}
\end{equation*}
where $Z$ denotes the privacy loss random variable of the DP kernel, and the probability of this leakage is:
\begin{equation}
\begin{aligned}
    W = &\text{Pr}(||N||\le\theta, ||N'||>\theta)\\
    =& \max\{(\Phi_{(\epsilon_y,\delta_y)}(\tau'_l)-\Phi_{(\epsilon_y,\delta_y)}(\tau_l)),\\&~~~~~~~~~~~~(\Phi_{(\epsilon_y,\delta_y)}(\tau'_u)-\Phi_{(\epsilon_y,\delta_y)}(\tau_u))\}.
\end{aligned}
\end{equation}

For case 1 and case 3:
\begin{equation*}
\begin{aligned}
    &\log\left\{\frac{\text{Pr}(\mathcal{M}(Q(X))=y_{n})}{\text{Pr}(\mathcal{M}(Q(X'))=y_{n})}\right\}
    =&Z,
\end{aligned}
\end{equation*}
and the probability of this leakage is $1-W$.
This concludes the proof for Theorem 1.

\end{proof}

\section{Proof the privacy profile}
\begin{proof}

\begin{equation}\label{eq:delta'}
    \begin{aligned}
    \delta_{\Gamma} \ge &\mathbb{E}_{\Gamma}[\max\{0,1-\exp(\epsilon-\gamma)\}]\\
    =&\int_{\epsilon}^{\infty}(1-\exp(\epsilon-\gamma))f_\Gamma(\gamma)d\gamma\\
    =&(1-W)\int_{\epsilon}^{\infty}(1-\exp(\epsilon-\gamma))f_Z(\gamma)d\gamma\\
    &~~~~~+W\int_{\epsilon}^{\infty}(1-\exp(\epsilon-\gamma))f_Z(\gamma-\mathcal{L})d\gamma\\
    =&(1-W)\int_{\epsilon}^{\infty}(1-\exp(\epsilon-z))f_Z(z)dz\\
    &~~~~~+W\int_{\epsilon-\mathcal{L}}^{\infty}(1-\exp(\epsilon-\mathcal{L}-z))f_Z(z)dz\\
    =&(1-W)\delta_Z(\epsilon)+W\delta_Z(\epsilon-\mathcal{L}).
    \end{aligned}
\end{equation}
    
\end{proof}

\section{Proof of Theorem 2}
\begin{proof}
Considering the three cases described in Appendix C.  For case 1 and 3: 

\begin{equation*}
\begin{aligned}
    \log\left\{\frac{\text{Pr}(\mathcal{M}(Q(X))=y_{n})}{\text{Pr}(\mathcal{M}(Q(X'))=y_{n})}\right\}
    = &Z,
\end{aligned}
\end{equation*}

For case 2:

\begin{equation*}
\begin{aligned}
    \log\left\{\frac{\text{Pr}(\mathcal{M}(Q(X))=y_{n})}{\text{Pr}(\mathcal{M}(Q(X'))=y_{n})}\right\}
    = Z + \log\left\{\frac{1}{{1-q}}\right\},
\end{aligned}
\end{equation*}

Combine these two cases:
\begin{equation*}
\begin{aligned}
    &\log\left\{\frac{\text{Pr}(\mathcal{M}(Q(X))=y_{n})}{\text{Pr}(\mathcal{M}(Q(X'))=y_{n})}\right\}
    \le Z + \log\left\{\frac{1}{{1-q}}\right\},
\end{aligned}
\end{equation*}

When the DP kernel the BR-DP framework satisfies $(\epsilon_y,\delta)$-DP, the following holds:
\begin{equation*}
    \text{Pr}\left\{\log\left(\frac{f_N(y_{n}|y)}{f_N(y_{n}|y')}\right)\ge\epsilon_y\right\}\le\delta,
\end{equation*}

\begin{equation*}
    \begin{aligned}
&\text{Pr}\left\{Z-\log{(1-q)}\le\epsilon_y-\log(1-q)\right\}\le{\delta},\\
\end{aligned}
\end{equation*}
which implies:
\begin{equation*}
    \text{Pr}\left\{\Gamma\le\epsilon_y-\log(1-q)\right\}\le{\delta}.\\
\end{equation*}
To guarantee $(\epsilon,\delta)$-DP:
and $q\le 1- \exp(\epsilon_y-\epsilon)$.
\end{proof}

\section{Proof of Theorem 4}
\begin{proof}
By definition, the privacy loss distribution,
\begin{equation*}
\begin{aligned}
&\Tilde{f}_{\Gamma}\left(\log\left(\frac{\text{Pr}(\mathcal{M}_0\cdot\mathcal{M}_1(X)=(y_0,y_1))}{\text{Pr}(\mathcal{M}_0\cdot\mathcal{M}_1(X')=(y_0,y_1))}\right)\right)\\
=&\text{Pr}(\mathcal{M}_0\cdot\mathcal{M}_1(X)=(y_0,y_1)).
\end{aligned}
\end{equation*}
Due to the independence of $\mathcal{M}_0$ and $\mathcal{M}_1$
\begin{equation*}
\begin{aligned}
    &\text{Pr}(\mathcal{M}_0\cdot\mathcal{M}_1(X)=(y_0,y_1))\\=&\text{Pr}(\mathcal{M}_0(X)=(y_0))\text{Pr}(\mathcal{M}_1(X)=(y_1)),
\end{aligned}
\end{equation*}
similarly, 
\begin{equation*}
\begin{aligned}
    &\text{Pr}(\mathcal{M}_0\cdot\mathcal{M}_1(X')=(y_0,y_1))\\=&\text{Pr}(\mathcal{M}_0(X')=(y_0))\text{Pr}(\mathcal{M}_1(X')=(y_1)),
\end{aligned}
\end{equation*}
Therefore, 
\begin{equation*}
\begin{aligned}
    &\log\left(\frac{\text{Pr}(\mathcal{M}_0\cdot\mathcal{M}_1(X)=(y_0,y_1))}{\text{Pr}(\mathcal{M}_0\cdot\mathcal{M}_1(X')=(y_0,y_1))}\right)\\
    =&\log\left(\frac{\text{Pr}(\mathcal{M}_0(X)=y_0)}{\text{Pr}(\mathcal{M}_0(X')=y_0)}\right)+\log\left(\frac{\text{Pr}(\mathcal{M}_1(X)=y_1)}{\text{Pr}(\mathcal{M}_1(X')=y_1)}\right),
\end{aligned}
\end{equation*}
and 
\begin{equation*}
\begin{aligned}
&\Tilde{f}_{\Gamma}\left(\log\left(\frac{\text{Pr}(\mathcal{M}_0(X)=y_0)}{\text{Pr}(\mathcal{M}_0(X')=y_0)}\right)+\log\left(\frac{\text{Pr}(\mathcal{M}_1(X)=y_1)}{\text{Pr}(\mathcal{M}_1(X')=y_1)}\right)\right)\\
=&\text{Pr}(\mathcal{M}_0(X)=(y_0))\text{Pr}(\mathcal{M}_1(X)=(y_1)).
\end{aligned}
\end{equation*}
which implies that 
\begin{equation*}
\begin{aligned}
    \Tilde{f}_{\Gamma}(\gamma)=&f^0_{\Gamma}(\gamma)\ast f^1_{\Gamma}(\gamma)\\
    =&f_Z^0(\gamma)\ast f^0_{R}(\gamma)\ast f_Z^1(\gamma)\ast f^1_R{(\gamma)}.
\end{aligned}
\end{equation*}
For independent and identical mechanisms, after $T$ compositions:
\begin{equation*}
\begin{aligned}
    \Tilde{f}_{\Gamma}(\gamma)=&f_{\Gamma}(\gamma)\ast^T f_{\Gamma}(\gamma)\\
    =&[(f_{Z}\ast^T f_{Z}\ast(f_{R}\ast^T f_{R})](\gamma),
\end{aligned}
\end{equation*}
where 
\begin{equation*}
    f_{R}\ast^T f_{R}(\gamma)=\sum_{k=0}^T\binom{T}{k} (1-W)^{k}W^{T-k}\delta_{Dirac}(\epsilon-(T-k)\mathcal{L}).
\end{equation*}
This completes the proof.
\end{proof}

\begin{figure*}[htp]
\centering 
\subfigure[Summation query for Gowalla ($\Delta_f = 22, \theta = 5$)]
{\includegraphics[width=0.43\textwidth]{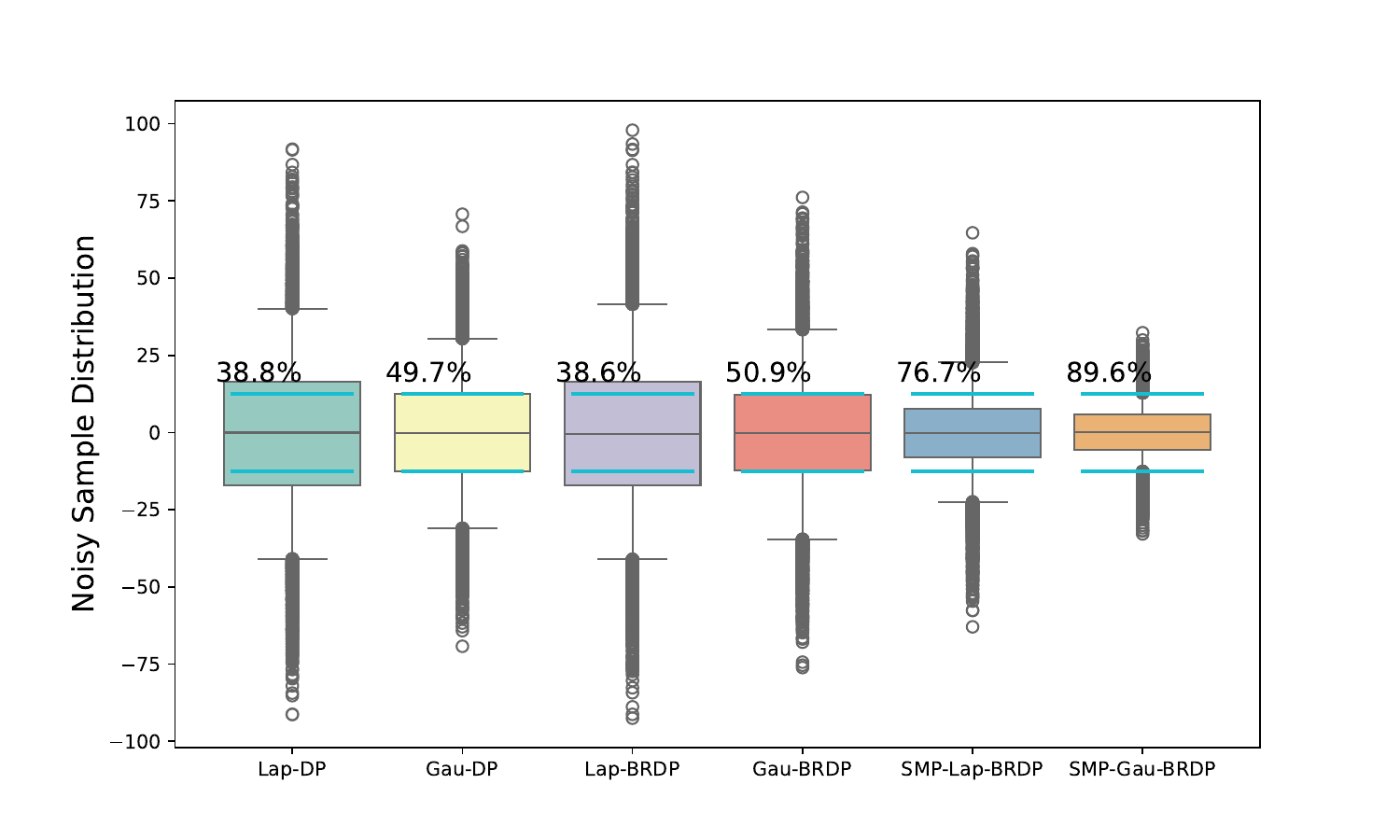}
\label{summ1}}
\subfigure[Summation query for Gowalla ($\Delta_f = 12, \theta = 10$)]
{\includegraphics[width=0.39\textwidth]{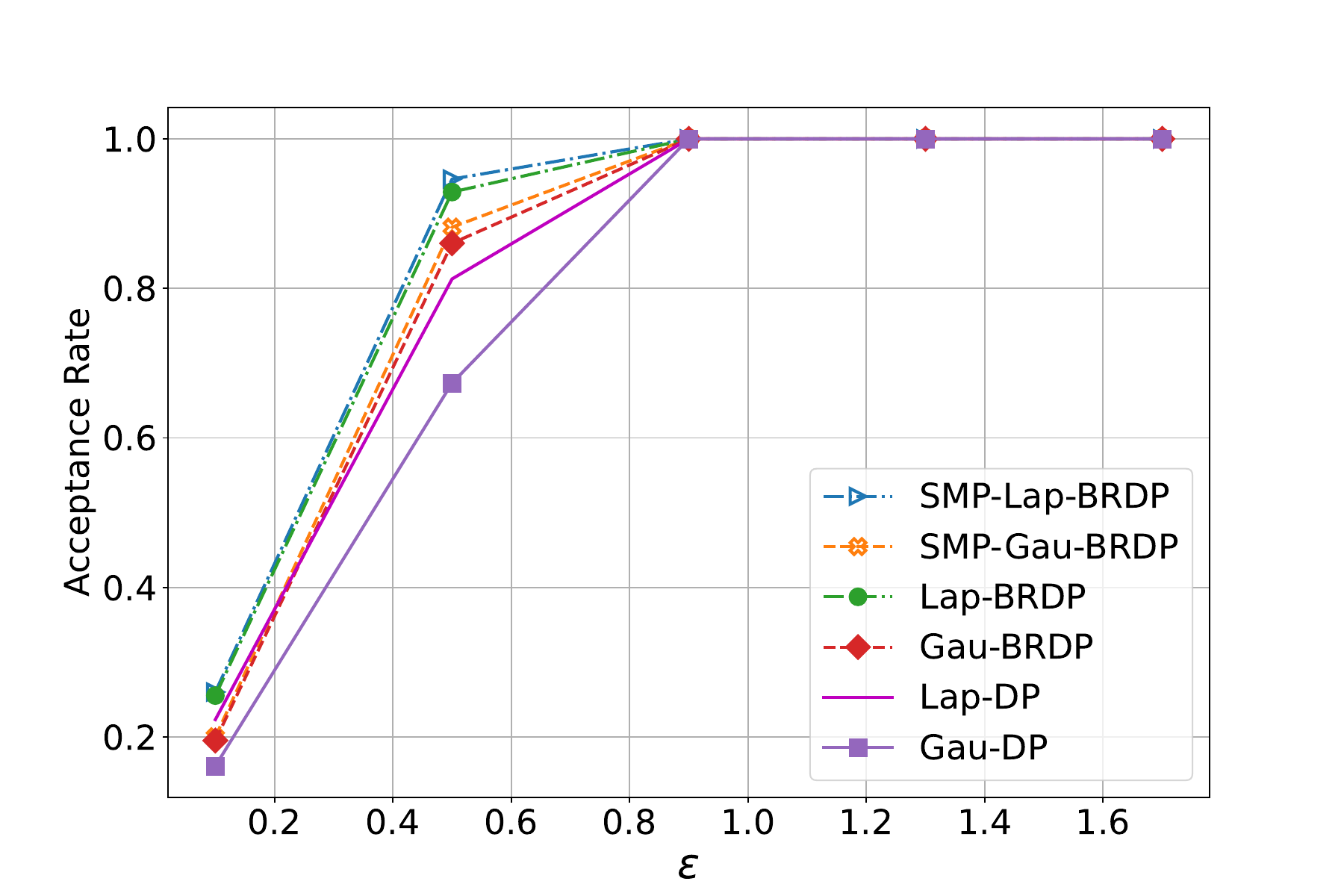}
\label{summ2}}
\subfigure[Averaging query for Gowalla ($\Delta_f = 0.0017,\theta =1$)]
{\includegraphics[width=0.43\textwidth]{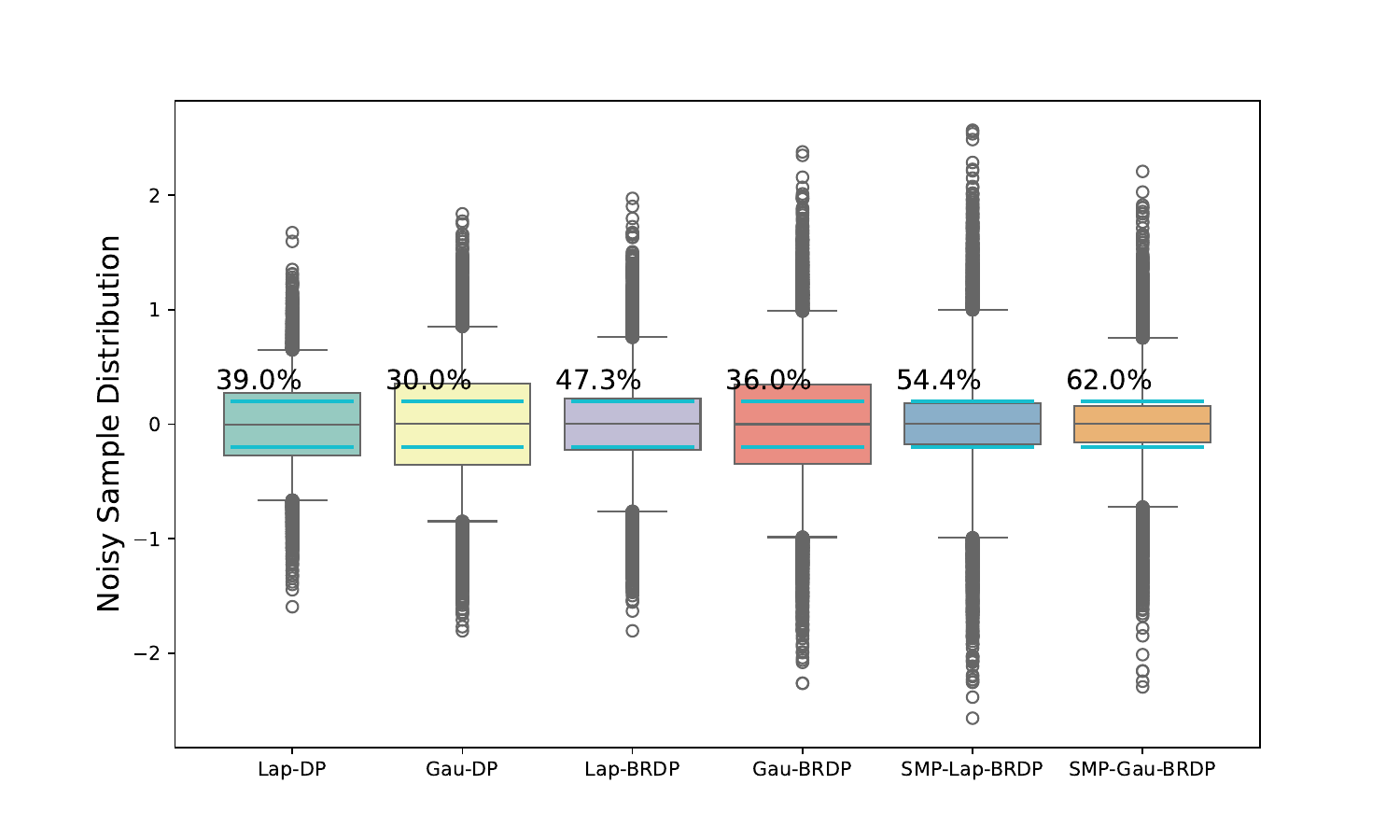}
\label{avg1}}
\subfigure[Averaging query for Gowalla ($\Delta_f = 0.0012,\theta =1$)]
{\includegraphics[width=0.39\textwidth]{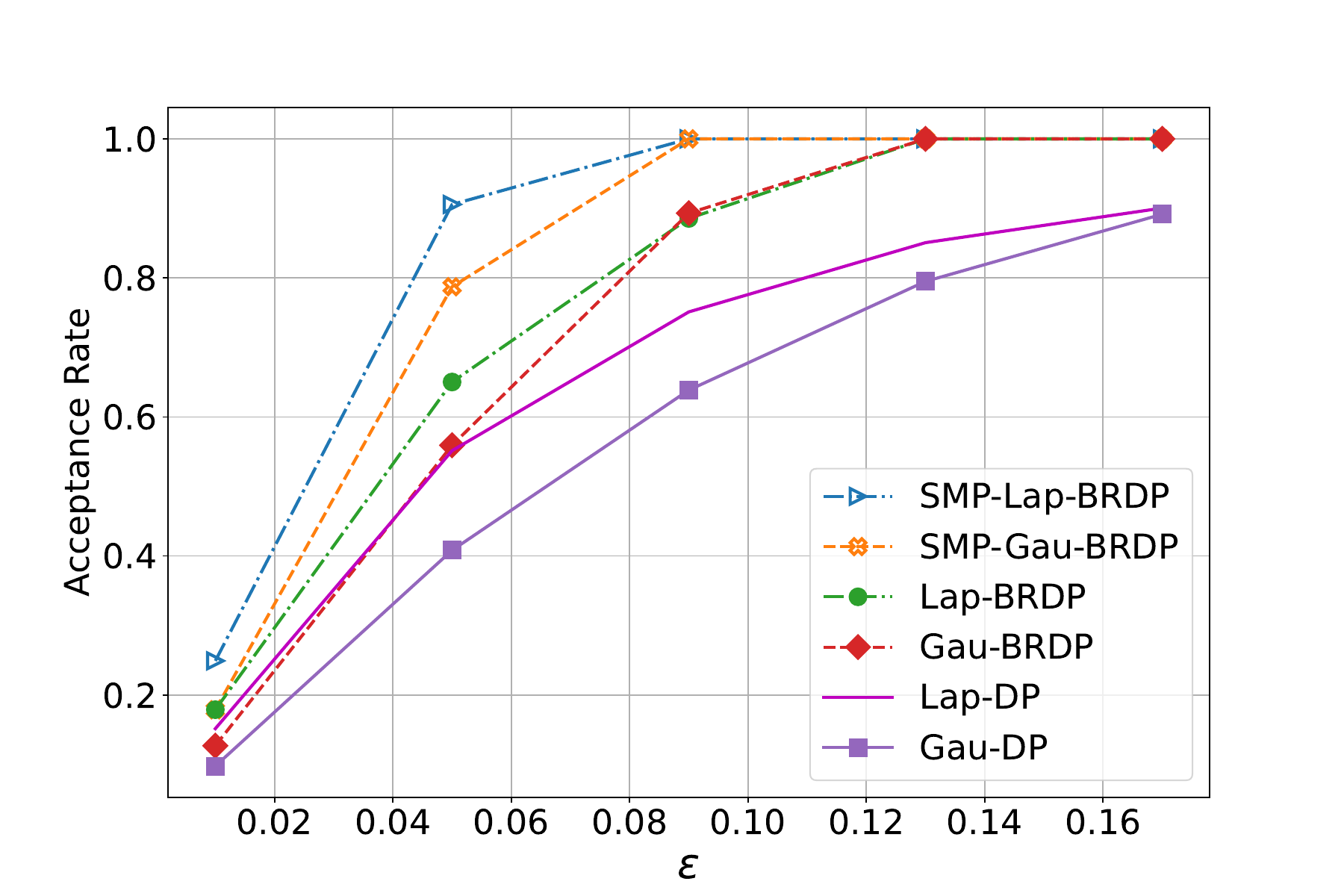}
\label{avg2}}
\subfigure[Counting query for Gowalla ($\Delta_f = 1, \theta=3$)]
{\includegraphics[width=0.43\textwidth]{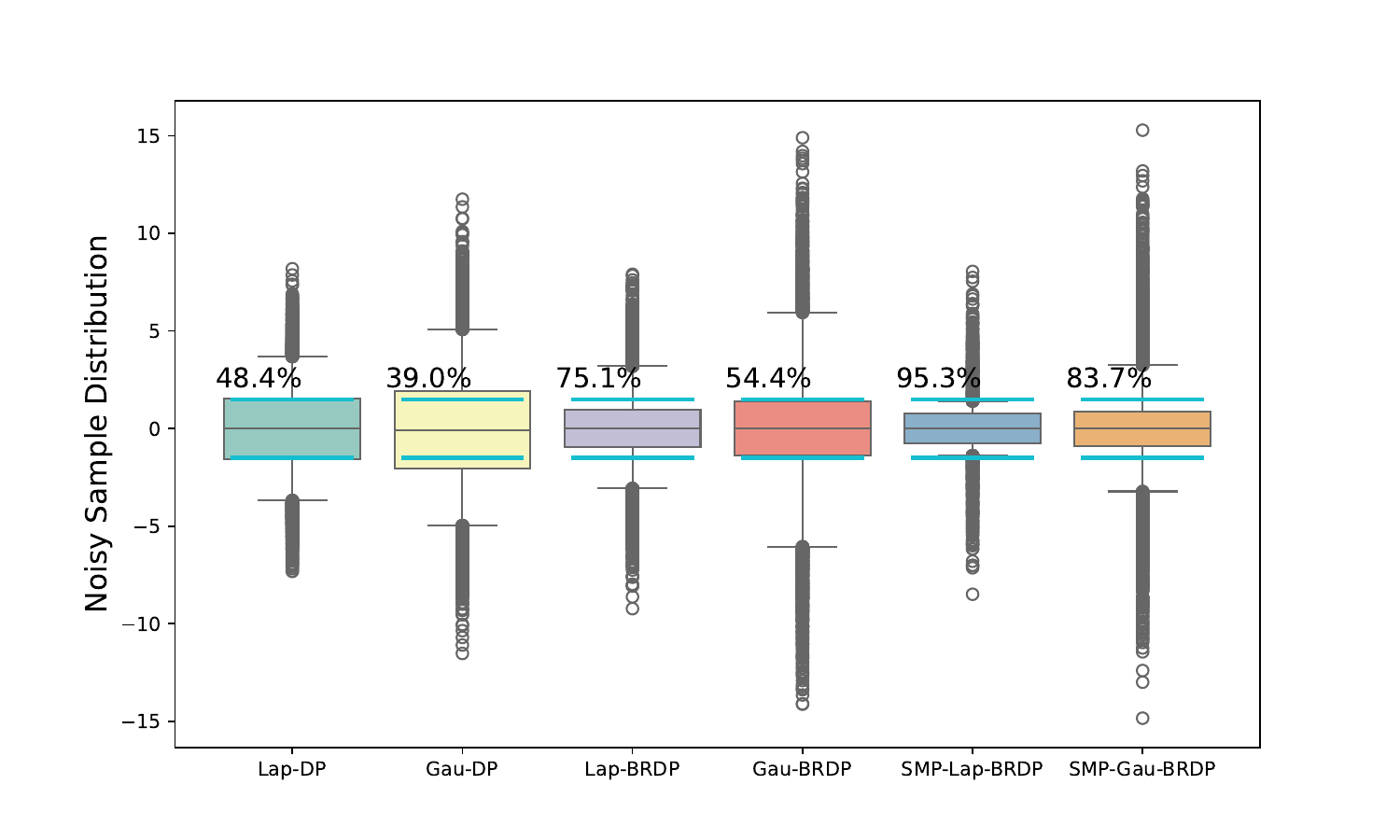}
\label{cnt1}}
\subfigure[Counting query for Gowalla ($\Delta_f = 1, \theta = 3$)]
{\includegraphics[width=0.39\textwidth]{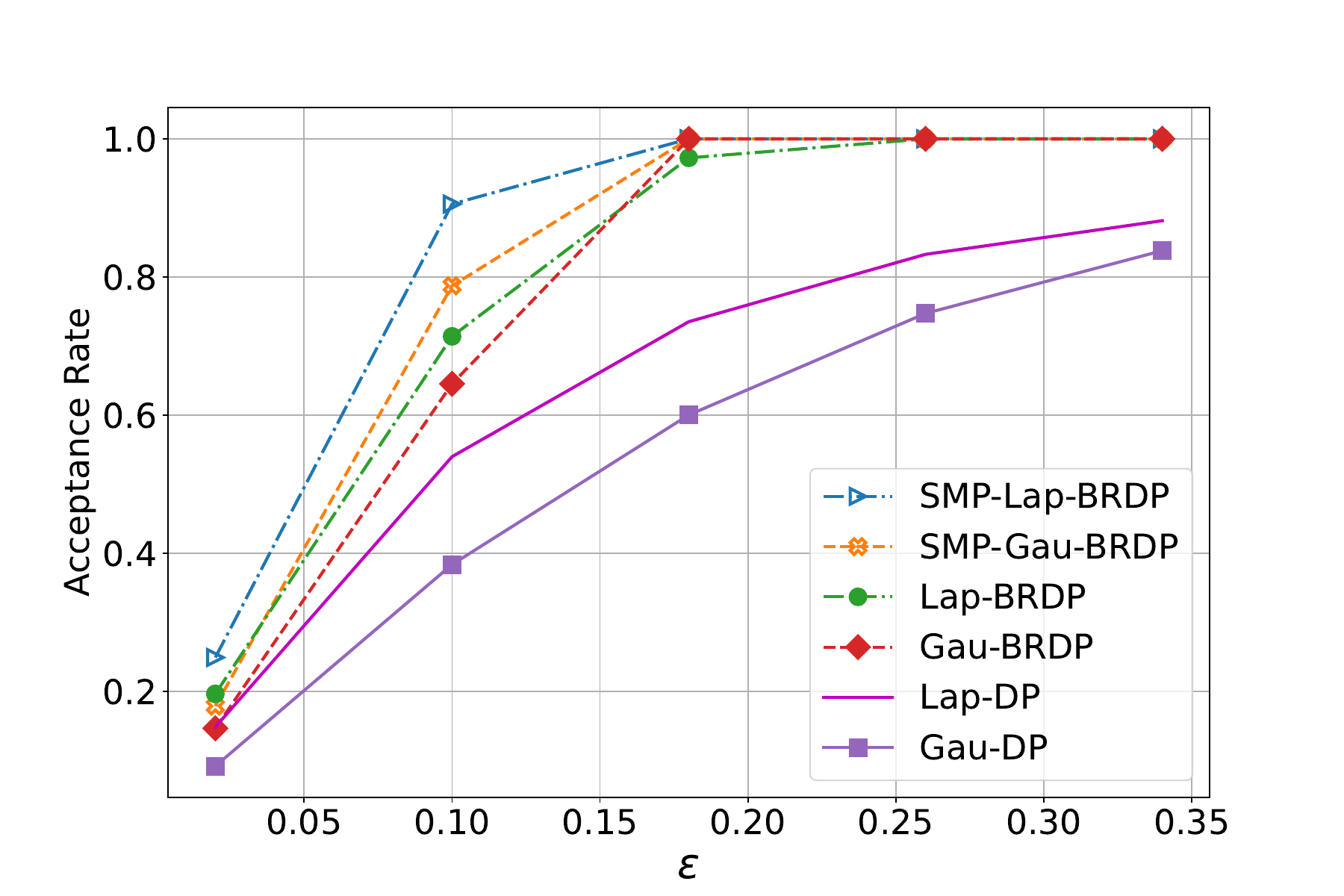}
\label{cnt2}}
\caption{Backup experiments with Gowalla: noisy data generations and utility-privacy tradeoff for different mechanisms under different queries for Gowalla datasets. For summation query $(\epsilon=0.5, \delta=10^-5)$, for averaging query $(\epsilon=0.05, \delta=10^-5)$, and for counting query $(\epsilon=0.1, \delta=10^-5)$.}
\label{fig:exp1}
\end{figure*}

\section{Proof for proposition 7} 

Expressions for different queries and subsampled aggregation are as follows:

\textbf{Summation query: } The summation query, denoted as $Q_{sum}$ can be expressed as:
\begin{equation*}
    Y = Q_{sum}(X) = \sum_{i\in{\mathcal{T}}}x_i,
\end{equation*}
and the aggregation result based on the subsampled dataset becomes:
\begin{equation*}
    \hat{Y}_{sum} = \frac{|X|}{|X_s|}Q_{sum}(X_s) = \frac{1}{p}\sum_{i\in{\mathcal{S}}}x_i.
\end{equation*}
and the noise introduced in the subsampling phase becomes:

\begin{equation*}
    \mathcal{E}_{sum} = \sum_{i\in{\mathcal{T}}}x_i - \frac{1}{p}\sum_{i\in{\mathcal{S}}}x_i.
\end{equation*}
We next derive the mean and variance for $\mathcal{E}_{sum}$, 
    \begin{equation*}
    \begin{aligned}
        E[\mathcal{E}_{sum}] = &E\left[\sum_{i\in\mathcal{T}}x_i\right] - \frac{1}{p}E\left[\sum_{i\in\mathcal{S}}x_i\right]\\
        =&|X|\mu - \frac{|X_s|}{p}\mu
        =0,
    \end{aligned}      
    \end{equation*}
and 
\begin{equation*}
    \mathcal{E}_{sum} = \sum_{i\in{\bar{\mathcal{S}}}} x_i +\left(1-\frac{1}{p}\right)\sum_{i\in\mathcal{S}}x_i,
\end{equation*}
\begin{equation*}
\begin{aligned}
    \text{Var}(\mathcal{E}_{sum}) = &\text{Var}\left[\sum_{i\in{\bar{\mathcal{S}}}} x_i\right] + \left(1-\frac{1}{p}\right)^2\text{Var}\left[\sum_{i\in\mathcal{S}}x_i\right]\\
    =&(1-p)|X|\text{Var}[x_i] + \frac{(1-p)^2}{p}\text{Var}[x_i]\\
    =&\sigma_x^2|X|\left(\frac{1-p}{p}\right).
\end{aligned}
\end{equation*}

\textbf{Average query:} For averaging query, the true aggregation and the estimated aggregation from a subsampled dataset becomes:

\begin{equation*}
    Y = Q_{avg}(X) = \frac{1}{|X|}\sum_{i\in{\mathcal{T}}}x_i,
\end{equation*}
the estimator of $Y$ with $X_s$ becomes
\begin{equation*}
    \hat{Y}_{avg} = \frac{1}{|X_s|}Q_{avg}(X_s) = \frac{1}{|X_s|}\sum_{i\in{\mathcal{S}}}x_i.
\end{equation*}
and the noise introduced in the subsampling phase can be viewed as:

\begin{equation*}
    \mathcal{E}_{avg} = \frac{1}{|X|}\sum_{i\in{\mathcal{T}}}x_i - \frac{1}{|X_s|}\sum_{i\in{\mathcal{S}}}x_i.
\end{equation*}
Mean of $\mathcal{E}_{avg}$:
\begin{equation*}
\begin{aligned}
E[\mathcal{E}_{avg}]=&\frac{1}{|X|}E\left[\sum_{i\in\mathcal{T}}x_i\right] - \frac{1}{|X_s|}E\left[\sum_{i\in\mathcal{S}}x_i\right]\\
& = \mu - \mu  =0.
\end{aligned}
\end{equation*}
For the variance of $\mathcal{E}_{avg}$:

\begin{equation*}
    \mathcal{E}_{avg} = \left(\frac{1}{|X|}-\frac{1}{|X_{\mathcal{S}}|}\right)\sum_{i\in\mathcal{S}}x_i + \frac{1}{|X|}\sum_{i\in\bar{\mathcal{S}}}x_i,
\end{equation*}
then,
\begin{small}

\begin{equation*}
\begin{aligned}
    \text{Var}(\mathcal{E}_{avg}) = &\left(\frac{1}{|X|}-\frac{1}{|X_{\mathcal{S}}|}\right)^2|X_{\mathcal{S}}|\text{Var}\left[x_i\right] + \frac{|X|-|X_{{\mathcal{S}}}|}{|X_s|^2}\text{Var}\left[x_i\right]\\
    = &\frac{\sigma_x^2}{|X|}\left(\frac{1-p}{p}\right).
\end{aligned}
\end{equation*}
    
\end{small}

\textbf{Counting query:} For counting query, denote a subset of the support of each $x_i$ as $\mathcal{C}$, and 
\begin{equation*}
    Y = Q_{cnt}(X) = \sum_{i\in{\mathcal{T}}} \mathbbm{1}_{\{x_i\in \mathcal{C}\} },
\end{equation*}
 and 
 \begin{equation*}
    \hat{Y}_{cnt} = \frac{|X|}{|X_s|}Q_{cnt}(X_s) = \frac{|X|}{|X_s|}\sum_{i\in{\mathcal{S}}} \mathbbm{1}_{\{x_i\in \mathcal{C}\} },
\end{equation*}
Therefore,
\begin{equation*}
    \mathcal{E}_{cnt} = \sum_{i\in{\mathcal{T}}} \mathbbm{1}_{\{x_i\in \mathcal{C}\}} - \frac{|X|}{|X_s|}\sum_{i\in{\mathcal{S}}} \mathbbm{1}_{\{x_i\in \mathcal{C}\} }.
\end{equation*}

Then for mean $\mathcal{E}_{cnt}$:

\begin{equation}
    \begin{aligned}
    E[\mathcal{E}_{cnt}] = |X|p_c - \frac{|X_s|p_c}{p} = 0,
    \end{aligned}
\end{equation}
similar to $\mathcal{E}_{cnt}$,  

\begin{equation*}
\begin{aligned}
    \text{Var}(\mathcal{E}_{cnt}) 
    =&p_c(1-p_c)|X|\left(\frac{1-p}{p}\right).
\end{aligned}
\end{equation*}
Note that binomial distribution under a large number of samples converges to a Gaussian distribution.

\section{Backup Experiment with Gowalla}
This section described backup experimental results with the dataset Gowalla, with similar settings as presented in the experiment section. Results are presented in Fig. \ref{fig:exp1}.

\end{document}